\documentclass[11pt,english]{article}
\usepackage{lmodern}

\usepackage[T1]{fontenc}
\usepackage[latin9]{inputenc}
\usepackage{amsmath}
\usepackage{amsthm}
\usepackage{amssymb}
\usepackage{geometry}
\makeatletter
\numberwithin{figure}{section}
\numberwithin{equation}{section}

\usepackage{float,psfrag,epsfig,color,url,hyperref}
\usepackage{algorithm,algorithmic}
\usepackage{graphicx,relsize}
\usepackage{amssymb,amsfonts,amsmath,amsthm,amscd,dsfont,mathrsfs,mathtools,microtype,nicefrac,pifont}
\usepackage{upgreek}
\usepackage[dvipsnames]{xcolor}
\usepackage{epstopdf,bbm,enumitem}
\usepackage{dsfont,tikz}
\usepackage[mathscr]{euscript}
\usepackage[toc,page]{appendix}
\hypersetup{
  colorlinks,
  linkcolor={red!50!black},
  citecolor={blue!50!black},
  urlcolor={blue!80!black}
}
\usepackage{etoolbox}
\patchcmd{\thebibliography}
  {\settowidth}
  {\setlength{\itemsep}{0pt plus 0.1pt}\settowidth}
  {}{}
\apptocmd{\thebibliography}
  {\small}
  {}{}
\allowdisplaybreaks
\usepackage{custom}
\makeatletter
\renewcommand{\paragraph}{%
  \@startsection{paragraph}{4}%
  {\z@}{1.25ex \@plus 1ex \@minus .2ex}{-1em}%
  {\normalfont\normalsize\bfseries}%
}
\makeatother
\makeatother

\theoremstyle{plain}
\newtheorem{thm}{\protect\theoremname}[section]
\newtheorem{cor}[thm]{\protect\corollaryname}
\newtheorem{lem}[thm]{\protect\lemmaname}
\newtheorem{prop}[thm]{\protect\propositionname}
\theoremstyle{remark}
\newtheorem*{acknowledgement*}{\protect\acknowledgementname}
\usepackage{babel}
\providecommand{\acknowledgementname}{Acknowledgement}
\providecommand{\corollaryname}{Corollary}
\providecommand{\lemmaname}{Lemma}
\providecommand{\propositionname}{Proposition}
\providecommand{\theoremname}{Theorem}

\begin{document}
%--------------------------------------------------------------------------------------------------------------------------------
% Environment shortcuts
%--------------------------------------------------------------------------------------------------------------------------------
\def\balign#1\ealign{\begin{align}#1\end{align}}
\def\baligns#1\ealigns{\begin{align*}#1\end{align*}}
\def\balignat#1\ealign{\begin{alignat}#1\end{alignat}}
\def\balignats#1\ealigns{\begin{alignat*}#1\end{alignat*}}
\def\bitemize#1\eitemize{\begin{itemize}#1\end{itemize}}
\def\benumerate#1\eenumerate{\begin{enumerate}#1\end{enumerate}}

% Align environments that use textstyle instead of displaystyle
\newenvironment{talign*}
 {\let\displaystyle\textstyle\csname align*\endcsname}
 {\endalign}
\newenvironment{talign}
 {\let\displaystyle\textstyle\csname align\endcsname}
 {\endalign}

\def\balignst#1\ealignst{\begin{talign*}#1\end{talign*}}
\def\balignt#1\ealignt{\begin{talign}#1\end{talign}}
%---------------------------------------------------

%--------------------------------------------------------------------------------------------------------------------------------
% Redefine left and right to remove initial and trailing space
%--------------------------------------------------------------------------------------------------------------------------------
\let\originalleft\left
\let\originalright\right
\renewcommand{\left}{\mathopen{}\mathclose\bgroup\originalleft}
\renewcommand{\right}{\aftergroup\egroup\originalright}

%--------------------------------------------------------------------------------------------------------------------------------
% Words with special symbols
%--------------------------------------------------------------------------------------------------------------------------------
\def\Gronwall{Gr\"onwall\xspace}
\def\Holder{H\"older\xspace}
\def\Ito{It\^o\xspace}
\def\Nystrom{Nystr\"om\xspace}
\def\Schatten{Sch\"atten\xspace}
\def\Matern{Mat\'ern\xspace}

%--------------------------------------------------------------------------------------------------------------------------------
% Smaller citations
%--------------------------------------------------------------------------------------------------------------------------------
\def\tinycitep*#1{{\tiny\citep*{#1}}}
\def\tinycitealt*#1{{\tiny\citealt*{#1}}}
\def\tinycite*#1{{\tiny\cite*{#1}}}
\def\smallcitep*#1{{\scriptsize\citep*{#1}}}
\def\smallcitealt*#1{{\scriptsize\citealt*{#1}}}
\def\smallcite*#1{{\scriptsize\cite*{#1}}}

%--------------------------------------------------------------------------------------------------------------------------------
% Colors
%--------------------------------------------------------------------------------------------------------------------------------
\def\blue#1{\textcolor{blue}{{#1}}}
\def\green#1{\textcolor{green}{{#1}}}
\def\orange#1{\textcolor{orange}{{#1}}}
\def\purple#1{\textcolor{purple}{{#1}}}
\def\red#1{\textcolor{red}{{#1}}}
\def\teal#1{\textcolor{teal}{{#1}}}

%--------------------------------------------------------------------------------------------------------------------------------
% Font styles
%--------------------------------------------------------------------------------------------------------------------------------
\def\mbi#1{\boldsymbol{#1}} % Bold and italic (math bold italic)
\def\mbf#1{\mathbf{#1}}
\def\mrm#1{\mathrm{#1}}
\def\tbf#1{\textbf{#1}}
\def\tsc#1{\textsc{#1}}

%--------------------------------------------------------------------------------------------------------------------------------
% Bold and italic variables
%--------------------------------------------------------------------------------------------------------------------------------
\def\mbiA{\mbi{A}}
\def\mbiB{\mbi{B}}
\def\mbiC{\mbi{C}}
\def\mbiDelta{\mbi{\Delta}}
\def\mbif{\mbi{f}}
\def\mbiF{\mbi{F}}
\def\mbih{\mbi{g}}
\def\mbiG{\mbi{G}}
\def\mbih{\mbi{h}}
\def\mbiH{\mbi{H}}
\def\mbiI{\mbi{I}}
\def\mbim{\mbi{m}}
\def\mbiP{\mbi{P}}
\def\mbiQ{\mbi{Q}}
\def\mbiR{\mbi{R}}
\def\mbiv{\mbi{v}}
\def\mbiV{\mbi{V}}
\def\mbiW{\mbi{W}}
\def\mbiX{\mbi{X}}
\def\mbiY{\mbi{Y}}
\def\mbiZ{\mbi{Z}}

%--------------------------------------------------------------------------------------------------------------------------------
% Textstyle vs. displaystyle
%--------------------------------------------------------------------------------------------------------------------------------
\def\textsum{{\textstyle\sum}} % Sum in textstyle form
\def\textprod{{\textstyle\prod}} % Prod in textstyle form
\def\textbigcap{{\textstyle\bigcap}} % Bigcap in textstyle form
\def\textbigcup{{\textstyle\bigcup}} % Bigcup in textstyle form

%--------------------------------------------------------------------------------------------------------------------------------
% Mathematical sets
%--------------------------------------------------------------------------------------------------------------------------------
\def\reals{\mathbb{R}} % Real number symbol
\def\integers{\mathbb{Z}} % Integer symbol
\def\rationals{\mathbb{Q}} % Rational numbers
\def\naturals{\mathbb{N}} % Natural numbers
\def\complex{\mathbb{C}} % Complex numbers

\def\what#1{\widehat{#1}}

\def\twovec#1#2{\left[\begin{array}{c}{#1} \\ {#2}\end{array}\right]}
\def\threevec#1#2#3{\left[\begin{array}{c}{#1} \\ {#2} \\ {#3} \end{array}\right]}
\def\nvec#1#2#3{\left[\begin{array}{c}{#1} \\ {#2} \\ \vdots \\ {#3}\end{array}\right]} % An n-vector with three arguments

%--------------------------------------------------------------------------------------------------------------------------------
% Eigenvalues
%--------------------------------------------------------------------------------------------------------------------------------
\def\maxeig#1{\lambda_{\mathrm{max}}\left({#1}\right)}
\def\mineig#1{\lambda_{\mathrm{min}}\left({#1}\right)}

%--------------------------------------------------------------------------------------------------------------------------------
% Operators
%--------------------------------------------------------------------------------------------------------------------------------
\def\Re{\operatorname{Re}} % Real part
\def\indic#1{\mbb{I}\left[{#1}\right]} % Indicator function
\def\logarg#1{\log\left({#1}\right)} % log with argument
\def\polylog{\operatorname{polylog}}
\def\maxarg#1{\max\left({#1}\right)} % max with argument
\def\minarg#1{\min\left({#1}\right)} % min with argument
\def\Earg#1{\E\left[{#1}\right]}
\def\Esub#1{\E_{#1}}
\def\Esubarg#1#2{\E_{#1}\left[{#2}\right]}
\def\bigO#1{\mathcal{O}\left(#1\right)} % big-oh notation
\def\littleO#1{o(#1)} % big-oh notation
\def\P{\mbb{P}} % Probability symbol
\def\Parg#1{\P\left({#1}\right)}
\def\Psubarg#1#2{\P_{#1}\left[{#2}\right]}
\def\Trarg#1{\Tr\left[{#1}\right]} % Trace with argument
\def\trarg#1{\tr\left[{#1}\right]} % trace with argument
\def\Var{\mrm{Var}} % Variance symbol
\def\Vararg#1{\Var\left[{#1}\right]}
\def\Varsubarg#1#2{\Var_{#1}\left[{#2}\right]}
\def\Cov{\mrm{Cov}} % Covariance symbol
\def\Covarg#1{\Cov\left[{#1}\right]}
\def\Covsubarg#1#2{\Cov_{#1}\left[{#2}\right]}
\def\Corr{\mrm{Corr}} % Covariance symbol
\def\Corrarg#1{\Corr\left[{#1}\right]}
\def\Corrsubarg#1#2{\Corr_{#1}\left[{#2}\right]}
\newcommand{\info}[3][{}]{\mathbb{I}_{#1}\left({#2};{#3}\right)} % Information symbol
\newcommand{\staticexp}[1]{\operatorname{exp}(#1)} % An exponential with parens that do not resize with input
\newcommand{\loglihood}[0]{\mathcal{L}} % log likelihood

% Copied from mathrsfs.sty

%--------------------------------------------------------------------------------------------------------------------------------
% Optimization macros
%--------------------------------------------------------------------------------------------------------------------------------
%\providecommand{\argmax}{\mathop\mathrm{arg max}} % Defining math symbols
%\providecommand{\argmin}{\mathop\mathrm{arg min}}
\providecommand{\arccos}{\mathop\mathrm{arccos}}
\providecommand{\dom}{\mathop\mathrm{dom}}
\providecommand{\diag}{\mathop\mathrm{diag}}
\providecommand{\tr}{\mathop\mathrm{tr}}
\providecommand{\card}{\mathop\mathrm{card}}
\providecommand{\sign}{\mathop\mathrm{sign}}
\providecommand{\conv}{\mathop\mathrm{conv}} % Convex hull
\def\rank#1{\mathrm{rank}({#1})}
\def\supp#1{\mathrm{supp}({#1})}

\providecommand{\minimize}{\mathop\mathrm{minimize}}
\providecommand{\maximize}{\mathop\mathrm{maximize}}
\providecommand{\subjectto}{\mathop\mathrm{subject\;to}}

\def\openright#1#2{\left[{#1}, {#2}\right)}

%--------------------------------------------------------------------------------------------------------------------------------
% Proof environments
%--------------------------------------------------------------------------------------------------------------------------------
\ifdefined\nonewproofenvironments\else
% The Theorems are numbered consecutively
% Lemmas are numbered by section, and observations, claims, facts, and 
% assumptions take their numbering. Propositions and definitions have their
% own numbering by section.
\ifdefined\ispres\else
% These conflict with Beamer definitions in pres mode
% \newtheorem{theorem}{Theorem}
% \newtheorem{lemma}[theorem]{Lemma}
% \newtheorem{corollary}[theorem]{Corollary}
% \newtheorem{definition}[theorem]{Definition}
% \newtheorem{fact}[theorem]{Fact}
% \renewenvironment{proof}{\noindent\textbf{Proof.}\hspace*{.3em}}{\qed \vspace{.1in}}
% \newenvironment{proof-sketch}{\noindent\textbf{Proof Sketch}
%   \hspace*{1em}}{\qed\bigskip\\}
% \newenvironment{proof-idea}{\noindent\textbf{Proof Idea}
%   \hspace*{1em}}{\qed\bigskip\\}
% \newenvironment{proof-of-lemma}[1][{}]{\noindent\textbf{Proof of Lemma {#1}}
%   \hspace*{1em}}{\qed\\}
%   \newenvironment{proof-of-proposition}[1][{}]{\noindent\textbf{Proof of Proposition {#1}}
%   \hspace*{1em}}{\qed\\}
% \newenvironment{proof-of-theorem}[1][{}]{\noindent\textbf{Proof of Theorem {#1}}
%   \hspace*{1em}}{\qed\\}
% \newenvironment{proof-attempt}{\noindent\textbf{Proof Attempt}
%   \hspace*{1em}}{\qed\bigskip\\}
% \newenvironment{proofof}[1]{\noindent\textbf{Proof of {#1}}
%   \hspace*{1em}}{\qed\bigskip\\}
 
% \newtheorem*{remark*}{Remark}
% \newenvironment{remark}{\noindent\textbf{Remark.}
%   \hspace*{0em}}{\smallskip}%\bigskip}
% \newenvironment{remarks}{\noindent\textbf{Remarks}
%   \hspace*{1em}}{\smallskip}
% \fi
% \newtheorem{observation}[theorem]{Observation}
% \newtheorem{proposition}[theorem]{Proposition}
% \newtheorem{claim}[theorem]{Claim}
% \newtheorem{assumption}{Assumption}
% \theoremstyle{definition}
% \newtheorem{example}[theorem]{Example}
%\renewcommand{\theassumption}{\Alph{assumption}} % Set counter for assumptions
                                                 % to be alphabetical
\fi
\fi
% Makes equation numbers have (1.1) style
% \numberwithin{equation}{section}
% \numberwithin{equation}{subsection}
\makeatletter
\@addtoreset{equation}{section}
\makeatother
\def\theequation{\thesection.\arabic{equation}}

\newcommand{\cmark}{\ding{51}}

\newcommand{\xmark}{\ding{55}}

%--------------------------------------------------------------------------------------------------------------------------------
% Equation environments
%--------------------------------------------------------------------------------------------------------------------------------
\newcommand{\eq}[1]{\begin{align}#1\end{align}}
\newcommand{\eqn}[1]{\begin{align*}#1\end{align*}}
\renewcommand{\Pr}[1]{\mathbb{P}\left( #1 \right)}
\newcommand{\Ex}[1]{\mathbb{E}\left[#1\right]}
%\newcommand{\var}[1]{\text{Var}\left(#1\right)}
%\newcommand{\ind}[1]{{\mathbbm{1}}_{\{ #1 \}} }
%\newcommand{\abs}[1]{\left|#1\right|}

%--------------------------------------------------------------------------------------------------------------------------------
% Comment environments
%--------------------------------------------------------------------------------------------------------------------------------
\newcommand{\matt}[1]{{\textcolor{Maroon}{[Matt: #1]}}}
\newcommand{\kook}[1]{{\textcolor{blue}{[Kook: #1]}}}
\definecolor{OliveGreen}{rgb}{0,0.6,0}
\newcommand{\sv}[1]{{\textcolor{OliveGreen}{[Santosh: #1]}}}

\global\long\def\on#1{\operatorname{#1}}%

\global\long\def\bw{\mathsf{Ball\ walk}}%
\global\long\def\sw{\mathsf{Speedy\ walk}}%
\global\long\def\gw{\mathsf{Gaussian\ walk}}%
\global\long\def\ps{\mathsf{Proximal\ sampler}}%
\global\long\def\dw{\mathsf{Dikin\ walk}}%

\global\long\def\chr{\mathsf{Coordinate\ Hit\text{-}and\text{-}Run}}%
\global\long\def\har{\mathsf{Hit\text{-}and\text{-}Run}}%
\global\long\def\gc{\mathsf{Gaussian\ cooling}}%
\global\long\def\ino{\mathsf{\mathsf{In\text{-}and\text{-}Out}}}%
\global\long\def\tgc{\mathsf{Tilted\ Gaussian\ cooling}}%
\global\long\def\PS{\mathsf{PS}}%
\global\long\def\psunif{\mathsf{PS}_{\textup{unif}}}%
\global\long\def\psexp{\mathsf{PS}_{\textup{exp}}}%
\global\long\def\psann{\mathsf{PS}_{\textup{ann}}}%
\global\long\def\psgauss{\mathsf{PS}_{\textup{Gauss}}}%
\global\long\def\eval{\mathsf{Eval}}%
\global\long\def\mem{\mathsf{Mem}}%
\global\long\def\HAR{\mathsf{HAR}}%
\global\long\def\CHAR{\mathsf{CHAR}}%
\global\long\def\HR{\msf{HR}}%

\global\long\def\O{O}%
\global\long\def\Otilde{\widetilde{O}}%
\global\long\def\Omtilde{\widetilde{\Omega}}%

\global\long\def\E{\mathbb{E}}%
\global\long\def\Z{\mathbb{Z}}%
\global\long\def\P{\mathbb{P}}%
\global\long\def\N{\mathbb{N}}%

\global\long\def\R{\mathbb{R}}%
\global\long\def\Rd{\mathbb{R}^{d}}%
\global\long\def\Rdd{\mathbb{R}^{d\times d}}%
\global\long\def\Rn{\mathbb{R}^{n}}%
\global\long\def\Rnn{\mathbb{R}^{n\times n}}%

\global\long\def\psd{\mathbb{S}^{d}_{+}}%
\global\long\def\pd{\mathbb{S}^{d}_{++}}%

\global\long\def\defeq{\stackrel{\mathrm{{\scriptscriptstyle def}}}{=}}%

\global\long\def\veps{\varepsilon}%
\global\long\def\lda{\lambda}%
\global\long\def\vphi{\varphi}%
\global\long\def\K{\mathcal{K}}%

\global\long\def\half{\frac{1}{2}}%
\global\long\def\nhalf{\nicefrac{1}{2}}%
\global\long\def\texthalf{{\textstyle \frac{1}{2}}}%
\global\long\def\ltwo{L^{2}}%

\global\long\def\ind{\mathds{1}}%
\global\long\def\op{\mathsf{op}}%
\global\long\def\ch{\mathsf{Ch}}%
\global\long\def\kls{\mathsf{KLS}}%
\global\long\def\ts{\mathsf{Ts}}%
\global\long\def\hs{\textup{HS}}%
\global\long\def\ls{\textup{LS}}%

\global\long\def\cpi{C_{\mathsf{PI}}}%
\global\long\def\cipi{C_{\mathsf{IPI}}}%
\global\long\def\clsi{C_{\mathsf{LSI}}}%
\global\long\def\cch{C_{\mathsf{Ch}}}%
\global\long\def\clch{C_{\mathsf{logCh}}}%
\global\long\def\cexp{C_{\mathsf{exp}}}%
\global\long\def\cgauss{C_{\mathsf{Gauss}}}%

\global\long\def\chooses#1#2{_{#1}C_{#2}}%

\global\long\def\frob{\on F}%

\global\long\def\vol{\on{vol}}%

\global\long\def\sym{\on{sym}}%

\global\long\def\law{\on{law}}%

\global\long\def\tr{\on{tr}}%

\global\long\def\diag{\on{diag}}%

\global\long\def\diam{\on{diam}}%

\global\long\def\poly{\on{poly}}%

\global\long\def\polylog{\on{polylog}}%

\global\long\def\Diag{\on{Diag}}%

\global\long\def\inter{\on{int}}%

\global\long\def\esssup{\on{ess\,sup}}%

\global\long\def\proj{\on{Proj}}%

\global\long\def\e{\mathrm{e}}%

\global\long\def\id{\mathrm{id}}%

\global\long\def\supp{\on{supp}}%

\global\long\def\spanning{\on{span}}%

\global\long\def\rows{\on{row}}%

\global\long\def\cols{\on{col}}%

\global\long\def\rank{\on{rank}}%

\global\long\def\T{\mathsf{T}}%

\global\long\def\bs#1{\boldsymbol{#1}}%

\global\long\def\eu#1{\EuScript{#1}}%

\global\long\def\mb#1{\mathbf{#1}}%

\global\long\def\mbb#1{\mathbb{#1}}%

\global\long\def\mc#1{\mathcal{#1}}%

\global\long\def\mf#1{\mathfrak{#1}}%

\global\long\def\ms#1{\mathscr{#1}}%

\global\long\def\mss#1{\mathsf{#1}}%

\global\long\def\msf#1{\mathsf{#1}}%

\global\long\def\textint{{\textstyle \int}}%
\global\long\def\Dd{\mathrm{D}}%
\global\long\def\D{\mathrm{d}}%
\global\long\def\grad{\nabla}%
 
\global\long\def\hess{\nabla^{2}}%
 
\global\long\def\lapl{\triangle}%
 
\global\long\def\deriv#1#2{\frac{\D#1}{\D#2}}%
 
\global\long\def\pderiv#1#2{\frac{\partial#1}{\partial#2}}%
 
\global\long\def\de{\partial}%
\global\long\def\lagrange{\mathcal{L}}%
\global\long\def\Div{\on{div}}%

\global\long\def\Gsn{\mathcal{N}}%
 
\global\long\def\BeP{\textnormal{BeP}}%
 
\global\long\def\Ber{\textnormal{Ber}}%
 
\global\long\def\Bern{\textnormal{Bern}}%
 
\global\long\def\Bet{\textnormal{Beta}}%
 
\global\long\def\Beta{\textnormal{Beta}}%
 
\global\long\def\Bin{\textnormal{Bin}}%
 
\global\long\def\BP{\textnormal{BP}}%
 
\global\long\def\Dir{\textnormal{Dir}}%
 
\global\long\def\DP{\textnormal{DP}}%
 
\global\long\def\Exp{\textnormal{Exp}}%
 
\global\long\def\Gam{\textnormal{Gamma}}%
 
\global\long\def\GEM{\textnormal{GEM}}%
 
\global\long\def\HypGeo{\textnormal{HypGeo}}%
 
\global\long\def\Mult{\textnormal{Mult}}%
 
\global\long\def\NegMult{\textnormal{NegMult}}%
 
\global\long\def\Poi{\textnormal{Poi}}%
 
\global\long\def\Pois{\textnormal{Pois}}%
 
\global\long\def\Unif{\textnormal{Unif}}%

\global\long\def\bpar#1{\bigl(#1\bigr)}%
\global\long\def\Bpar#1{\Bigl(#1\Bigr)}%

\global\long\def\abs#1{|#1|}%
\global\long\def\babs#1{\bigl|#1\bigr|}%
\global\long\def\Babs#1{\Bigl|#1\Bigr|}%

\global\long\def\snorm#1{\|#1\|}%
\global\long\def\bnorm#1{\bigl\Vert#1\bigr\Vert}%
\global\long\def\Bnorm#1{\Bigl\Vert#1\Bigr\Vert}%

\global\long\def\sbrack#1{[#1]}%
\global\long\def\bbrack#1{\bigl[#1\bigr]}%
\global\long\def\Bbrack#1{\Bigl[#1\Bigr]}%

\global\long\def\sbrace#1{\{#1\}}%
\global\long\def\bbrace#1{\bigl\{#1\bigr\}}%
\global\long\def\Bbrace#1{\Bigl\{#1\Bigr\}}%

\global\long\def\Abs#1{\left\lvert #1\right\rvert }%
\global\long\def\Par#1{\left(#1\right)}%
\global\long\def\Brack#1{\left[#1\right]}%
\global\long\def\Brace#1{\left\{  #1\right\}  }%

\global\long\def\inner#1{\langle#1\rangle}%
 
\global\long\def\binner#1#2{\left\langle {#1},{#2}\right\rangle }%

\global\long\def\norm#1{\lVert#1\rVert}%
\global\long\def\onenorm#1{\norm{#1}_{1}}%
\global\long\def\twonorm#1{\norm{#1}_{2}}%
\global\long\def\infnorm#1{\norm{#1}_{\infty}}%
\global\long\def\fronorm#1{\norm{#1}_{\text{F}}}%
\global\long\def\nucnorm#1{\norm{#1}_{*}}%
\global\long\def\staticnorm#1{\|#1\|}%
\global\long\def\statictwonorm#1{\staticnorm{#1}_{2}}%

\global\long\def\mmid{\mathbin{\|}}%

\global\long\def\otilde#1{\widetilde{O}(#1)}%
\global\long\def\wtilde{\widetilde{W}}%
\global\long\def\wt#1{\widetilde{#1}}%

\global\long\def\KL{\msf{KL}}%
\global\long\def\dtv{d_{\textrm{\textup{TV}}}}%
\global\long\def\FI{\msf{FI}}%
\global\long\def\tv{\msf{TV}}%
\global\long\def\TV{\msf{TV}}%

\global\long\def\cov{\on{cov}}%
\global\long\def\var{\on{Var}}%
\global\long\def\ent{\on{Ent}}%

\global\long\def\cred#1{\textcolor{red}{#1}}%
\global\long\def\cblue#1{\textcolor{blue}{#1}}%
\global\long\def\cgreen#1{\textcolor{green}{#1}}%
\global\long\def\ccyan#1{\textcolor{cyan}{#1}}%
\global\long\def\yk#1{\textcolor{red}{\textsf{[YK: #1]}}}%
\global\long\def\yb#1{\textcolor{blue}{\textsf{[yb: #1]}}}%

\global\long\def\iff{\Leftrightarrow}%
 
\global\long\def\textfrac#1#2{{\textstyle \frac{#1}{#2}}}%

\global\long\def\Expo{\textnormal{Expo}}%
\global\long\def\Tr{\on{Tr}}%
\global\long\def\onu{\bar{\nu}}%
\global\long\def\intk{\inter\K}%
\global\long\def\ncal{\mathcal{N}}%
\global\long\def\svec{\operatorname{svec}}%
\global\long\def\tvec{\operatorname{vec}}%
\global\long\def\del{\partial}%
\global\long\def\ovec{\operatorname{ovec}}%

\global\long\def\Sph{\mathbb{S}}%
\global\long\def\gap{\operatorname{gap}}%
\global\long\def\Id{\mathrm{Id}}%
\global\long\def\logp{\log_{+}}%

\title{Spectral Gaps of Hit-and-Run and Coordinate Hit-and-Run\author{Yunbum Kook\\ University of Michigan\\  \texttt{ybkook@umich.edu} \and Santosh S. Vempala\\ Georgia Tech\\ \texttt{vempala@gatech.edu}}}
\maketitle
\begin{abstract}
For any convex body $\mathcal{K}\subset\mathbb{R}^{n}$ containing
a unit ball, the spectral gap of Hit-and-Run is $\Omega(1/(n^{2}C_{\mathsf{PI}}))$,
where $C_{\mathsf{PI}}$ is the Poincar\'e constant of the uniform
distribution $\pi$ over $\mathcal{K}$. This implies that Hit-and-Run
converges to a distribution within $\chi^{2}$-divergence $\varepsilon$
of the uniform distribution $\pi$ in $O(n^{2}C_{\mathsf{PI}}\log(M/\varepsilon))$
steps from any starting distribution $\pi_{0}$ with $M=\chi^{2}(\pi_{0}\,\|\,\pi)$,
thus refining the known bound of $O(n^{2}R^{2}\log(M/\varepsilon))$
by Lov\'asz and Vempala (2004) in terms of the outer radius $R$;
for nearly isotropic bodies, together with progress on the KLS conjecture,
the complexity is $O(n^{2}\log n\log(M/\varepsilon))$, improving
the dimension dependence from cubic to nearly quadratic while maintaining
logarithmic dependence on the initial distance. It was an open problem
to connect the convergence of Hit-and-Run to Poincar\'e/KLS constants
as was done for the Ball walk by Kannan, Lov\'asz and Simonovits
(1997). Unlike Hit-and-Run, the Ball walk has an unavoidable linear
dependence on (a stronger notion) of the initial warmness.

We directly bound the spectral gap of the Hit-and-Run Markov chain
by connecting it to functional isoperimetric constants, inspired by
the recent analysis of In-and-Out. Rewriting the spectral gap first
in terms of dual certificates leads to the Babu\v{s}ka--Aziz constant
studied in the analysis of PDEs; it is asymptotically bounded by the
Improved Poincar\'e constant of the target distribution, which we
show can be bounded in terms of the usual Poincar\'e constant. The
proof is based on duality and calculus, unlike known proofs of convergence
for Hit-and-Run which are based on bounding the conductance. The same
technique can be applied to Coordinate Hit-and-Run, resulting in a
much improved mixing time of $O(n^{3}C_{\mathsf{PI}}\log(M/\varepsilon))$.
\end{abstract}
\vfill{}
\noindent\textbf{AI disclosure: }A preliminary version of the main technique/proof/figure
was suggested by GPT 5.6 Pro through guided prompting by the first
author. The authors then verified, streamlined and wrote the full
proof. 

\newpage{}

\section{Introduction}

Sampling a convex body is a fundamental algorithmic problem. Hit-and-Run,
proposed independently by Boneh and Golan \cite{BG79constraint} and
Smith \cite{Smith84HAR}, is the following Markov chain (see Figure~\ref{fig:har-char})
when applied to a convex body $\K\subset\R^{n}$: at a current point
$x\in\K$, 
\begin{enumerate}
\item Sample a uniformly distributed random line $\ell$ through $x$.
\item Go to a uniform random point on the chord $\ell\cap\K$. 
\end{enumerate}
$\har$ is an attractive, easy-to-implement process that has been
widely used in practice. Coordinate Hit-and-Run ($\CHAR$), also known
as Gibbs sampling and introduced in 1971 \cite{Turchin71computation},
is the variant that only uses axis-parallel chords and is also popular
due to its low memory overhead. 

Building on the seminal work of \cite{Lovasz99hit}, Lov\'asz and
Vempala \cite{LV06hit} showed that $\har$ mixes from a cold start:
in any convex body $\K\subset\R^{n}$ with $B^{n}\subseteq\K\subseteq RB^{n}$,
it outputs a point within $\chi^{2}$-divergence $\varepsilon$ of
the uniform distribution over $\K$ in $O(n^{2}R^{2}\log(M/\varepsilon))$
steps from any starting distribution within $\chi^{2}$-divergence
$M$ of the target. This convergence rate is asymptotically optimal
in terms of these parameters, as shown by a cylinder whose cross section
is a unit ball and axis has length $2R$. It also means that $\har$
mixes rapidly from any interior starting point (by applying the convergence
bound to the distribution obtained after taking one step), the first
random walk known to have this property for arbitrary convex bodies.
More recently, $\CHAR$ was shown to have (higher) polynomial convergence
rates \cite{LV23char,NS22char,NRS25sampling}, with the last paper
showing logarithmic dependence on the warm start parameter.

The study of efficiently sampling convex bodies has made much progress
over the past few decades, starting from an initial polynomial-time
algorithm with complexity $n^{23}$ by Dyer, Frieze, and Kannan \cite{DFK91random}.
The $\bw$ introduced by Lov\'asz \cite{lovasz90compute} has been
particularly fruitful, leading to the current best complexity for
sampling, general techniques for the analysis of Markov chains \cite{LS90mixing,LS93random,KLS97random,LV07geometry},
and novel classes of isoperimetric inequalities. One highlight is
the Kannan--Lov\'asz--Simonovits (KLS) conjecture \cite{KLS95isop}
which posits that halfspaces are asymptotically optimal isoperimetric
cuts for convex bodies (and logconcave densities). The conjecture
directly leads to a better bound on the convergence rate of the $\bw$,
and has several other interesting mathematical consequences and connections.
Another highlight is the localization method, first developed by Kannan,
Lov\'asz, and Simonovits \cite{LS93random,KLS95isop} to prove inequalities
in high dimension and later generalized to a stochastic method by
Eldan \cite{eldan13thin}. For a more comprehensive account of these
developments, we refer the reader to \cite{vempala05geometric,KV25localization}.

The $\bw$ has a polynomial dependence on the warmness parameter\footnote{In fact, the $\bw$ uses a more stringent notion of warmness, namely
$\D\pi_{0}/\D\pi\le M$.}, and this is unavoidable. To get around this, researchers have developed
algorithmic solutions, notably annealing \cite{LV06simulated,KV06simulated,CV18Gaussian},
to arrange an $O(1)$-warm start. From such a warm start, the $\bw$'s
convergence rate is roughly $n^{2}\,\norm{\cov}_{\op}\,\psi^{2}_{n}$,
where $\norm{\cov}_{\op}$ is the operator norm of the covariance
of the target distribution, and $\psi_{n}$ is the KLS constant, which
is known to be $O(\sqrt{\log n})$ and is conjectured to be $O(1)$.
This is always a better bound than $n^{2}R^{2}$ and is asymptotically
better by a factor of $n$ when the target is isotropic (i.e., $\norm{\cov}_{\op}=1)$.
Note that for convex bodies and logconcave densities, we have $\cpi\simeq\norm{\cov}_{\op}\psi^{2}_{n}$,
where the RHS is the squared reciprocal of the Cheeger constant.

On the other hand, despite its logarithmic dependence on the warmness
parameter, there was no known connection between the convergence of
$\har$ and the KLS constant. Thus, a decade of progress on the latter,
leading to improved complexities for sampling (near-)isotropic convex
bodies, did not refine the bound for $\har$. Unlike the $\bw$, whose
analysis directly uses Euclidean isoperimetry, the analysis of $\har$
was based on an average isoperimetric inequality for the cross-ratio
distance, and the latter was already the best possible. Thus, it remained
a tantalizing open problem to bridge the analysis of $\har$ with
improved isoperimetric constants. 

Chen and Eldan \cite{CE25hitandrun} made a conceptually important
breakthrough, showing that for isotropic convex bodies, the convergence
rate is indeed nearly quadratic in the dimension, albeit with a polynomial
dependence on initial warmness and the distance to the target. Interestingly,
their proof technique, called a localization scheme \cite{CE25localization},
used a stochastic localization process to reduce the analysis from
general convex bodies to that of a highly concentrated Gaussian restricted
to a convex body; this general stochastic technique has also been
the driving force for much of the progress in bounding the KLS constant.
So, although they lost the logarithmic dependence on the warmness
parameter, they showed that the conjectured dimension dependence is
indeed plausible and valid in the isotropic setting.

\begin{figure}[t]
\centering
% =========================
% Hit-and-Run
% =========================
\begin{minipage}{0.49\textwidth}
\centering
{\large\bfseries Hit-and-Run}
\vspace{4pt}
\begin{tikzpicture}[scale=1.15]
  % Convex polytope
  \filldraw[
    thick,
    fill=blue!8,
    draw=black!60
  ]
  (-1.95,-0.45) --
  (-1.45,0.90) --
  (-0.35,1.20) --
  (1.25,0.95) --
  (1.90,0.25) --
  (1.45,-0.90) --
  (0.25,-1.15) --
  (-1.35,-0.95) -- cycle;

  % Current point x and sampled point y
  \coordinate (x) at (-0.65,-0.25);
  \coordinate (y) at (0.65,0.35);

  % Unit sphere around x
  \draw[thin, black!60]
    (x) circle (0.48);

  % Random line ell extending beyond the body
  \draw[dashed, thick]
    (-2.40,-1.058) -- (2.40,1.158);

  % x label moved southeast
  \fill (x) circle (1.6pt)
    node[below right=2pt] {$x$};

  % y label moved northwest
  \fill (y) circle (1.6pt)
    node[above left=1pt] {$y$};

  % Line label
  \node[above right] at (1.5,1) {$\ell$};

  % Transition description
  \node at (0,-1.48)
    {$y\sim \operatorname{Unif}(\ell\cap\K)$};
\end{tikzpicture}
\end{minipage}
\hfill
% =========================
% Coordinate Hit-and-Run
% =========================
\begin{minipage}{0.49\textwidth}
\centering
{\large\bfseries Coordinate Hit-and-Run}
\vspace{4pt}
\begin{tikzpicture}[scale=1.15]
  % Convex polytope
  \filldraw[
    thick,
    fill=blue!8,
    draw=black!60
  ]
  (-1.95,-0.45) --
  (-1.45,0.90) --
  (-0.35,1.20) --
  (1.25,0.95) --
  (1.90,0.25) --
  (1.45,-0.90) --
  (0.25,-1.15) --
  (-1.35,-0.95) -- cycle;

  % Current point x and sampled point y
  \coordinate (x) at (-0.65,0);
  \coordinate (y) at (0.75,0);

  % Axis-parallel lines through x
  \draw[dashed, thick]
    (-2.40,0) -- (2.40,0);
  \draw[dashed, thick]
    (-0.65,-1.40) -- (-0.65,1.40);

  % Current and sampled points
  \fill (x) circle (1.6pt)
    node[below=3pt] {$x$};
  \fill (y) circle (1.6pt)
    node[below=3pt] {$y$};

  % Coordinate-line label
  \node[above=5pt] at (2.2,-0.2) {$\ell_i$};

  % Transition description
  \node at (0,-1.48)
    {$y\sim \operatorname{Unif}(\ell_i\cap\K)$};
\end{tikzpicture}
\end{minipage}
\caption{
\textbf{Left}: A HAR kernel at $x \in \K$ chooses a random line $\ell$ through $x$ and samples $y$ uniformly from the chord $\ell\cap\K$.
\textbf{Right}: A CHAR kernel chooses an axis-parallel line $\ell_i$ through $x$ and samples $y$ uniformly from the chord $\ell_i\cap\K$.
}
\label{fig:har-char}
\end{figure}
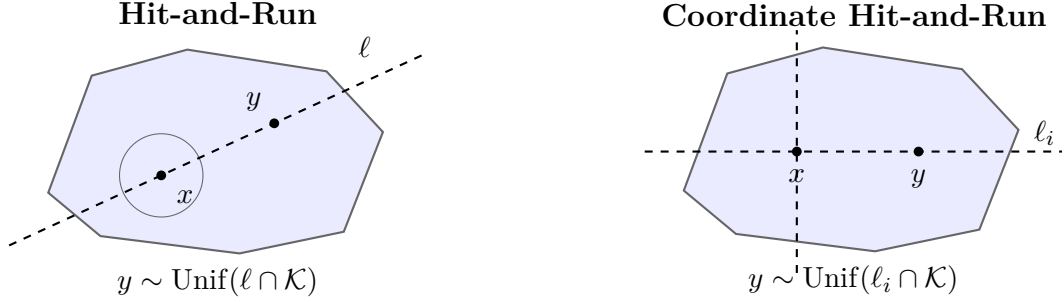

Recent work \cite{KVZ26INO,KV25sampling,KV26zeroLC} took a different
approach inspired by continuous diffusion, and showed that the $\ino$
random walk has a convergence rate that matches the best-known bounds
for the $\bw$, with stronger output guarantees and, importantly,
a novel proof framework that directly connects convergence (in fact,
per-step contraction towards the target) with functional isoperimetric
constants of the target distribution. Nevertheless, $\ino$ also has
an unavoidable linear dependence on the warmness parameter, highlighting
the open problem about $\har$:

\smallskip{}

\emph{Can the convergence rate of $\har$ be bounded by the Poincar\'e/KLS
constants while maintaining its logarithmic dependence on the warmness
parameter?}

\smallskip{}

In this paper, we answer this question affirmatively for both $\har$
and $\CHAR$. The proof technique bypasses the traditional method
of bounding the conductance and directly connects the convergence
rate to functional isoperimetric constants. 

The main high-level idea of the proof is to produce a dual certificate
bounding the spectral gap. Rewriting the spectral gap first in terms
of dual certificates leads to the Babu\v{s}ka--Aziz constant \cite{BA72survey,HP83inequalities}
studied in the analysis of PDEs; for bounded domains, this is known
to be asymptotically bounded by the Improved Poincar\'e constant
of the target density. The latter is a refined version of the classical
Poincar\'e constant, incorporating distance to the boundary as a
weight function for points in the domain \cite{HurriSyrjanen1994,Costabel2017,Zsuppan2020}.
We show that this constant can be bounded in terms of the usual Poincar\'e
constant for convex bodies resulting in a mixing time bound of $O(n^{2}\cpi\log(M/\varepsilon))$.
Applying this to $\CHAR$, with the only change being in the first
step of connecting the spectral gap to the Babu\v{s}ka--Aziz constant,
results in a mixing time bound of $O(n^{3}\cpi\log(M/\varepsilon))$,
which is conjectured to be asymptotically tight. We note that under
this new analytical framework, the proofs for $\har$ and $\CHAR$
are essentially the same.

We discuss the proof techniques in more detail after stating the main
results. 

\subsection{Results}

Below, $\cpi(\pi)$ denotes the Poincar\'e constant of a probability
measure $\pi$, formally defined in \S\ref{sec:prelim}.
\begin{thm}[Spectral gaps]
\label{thm:main} Let $n\geq2$ and $\K\subset\R^{n}$ be a bounded,
open, and convex set with unit ball inside, and let $\pi$ be the
uniform distribution over $\K$. Then, the spectral gaps of Hit-and-Run
and Coordinate Hit-and-Run satisfy
\[
\lda_{\HR}\gtrsim\frac{1}{n^{2}\cpi(\pi)}\gtrsim\frac{1}{n^{2}\,\norm{\cov\pi}_{\op}\log n}\,,\quad\lda_{\CHAR}\gtrsim\frac{1}{n^{3}\cpi(\pi)}\gtrsim\frac{1}{n^{3}\,\norm{\cov\pi}_{\op}\log n}\,.
\]
\end{thm}

\begin{cor}
[Mixing time]\label{cor:mixing-har-char} Under the assumptions of
Theorem~\ref{thm:main}, let $\pi^{\HR}_{N}$ and $\pi^{\CHAR}_{N}$
be the distribution of the $N$-th iterate of Hit-and-Run and Coordinate
Hit-and-Run initialized at $\pi_{0}$, respectively. Then, for some
universal constant $c>0$
\begin{align*}
\chi^{2}(\pi^{\HR}_{N}\mmid\pi) & \leq\exp\bpar{-\frac{cN}{n^{2}\,\norm{\cov\pi}_{\op}\log n}}\,\chi^{2}(\pi_{0}\mmid\pi)\,,\\
\chi^{2}(\pi^{\CHAR}_{N}\mmid\pi) & \leq\exp\bpar{-\frac{cN}{n^{3}\,\norm{\cov\pi}_{\op}\log n}}\,\chi^{2}(\pi_{0}\mmid\pi)\,.
\end{align*}
In particular, for given $\veps>0$, the number of iterations for
Hit-and-Run and Coordinate Hit-and-Run to achieve $\veps$-distance
to $\pi$ in $\chi^{2}$-divergence is bounded as
\[
N_{\HR}\lesssim n^{2}\,\norm{\cov\pi}_{\op}\log n\log\frac{\chi^{2}(\pi_{0}\mmid\pi)}{\veps}\,,\qquad N_{\CHAR}\lesssim n^{3}\,\norm{\cov\pi}_{\op}\log n\log\frac{\chi^{2}(\pi_{0}\mmid\pi)}{\veps}\,.
\]
\end{cor}

\subsection{Technical overview\label{sec:prelim}}

\subsubsection{Preliminaries}

We recall the functional-analytic facts used later, with more details
in \S\ref{sec:Functional-analytic-background}. 

\paragraph{Basics. }

Throughout the paper, we use the same symbol to denote a probability
measure and its Lebesgue density when it is clear from context, and
$\D\pi(x)=\frac{\D x}{\vol\K}$ denotes the uniform measure on $\K\subset\Rn$.
We write $\inner{f,g}_{\pi}:=\int_{\K}fg\,\D\pi$, $\norm f^{2}_{2}:=\inner{f,f}_{\pi}=\E_{\pi}[f^{2}]$,
and $L^{2}_{0}(\pi):=\{f\in L^{2}(\pi):\int_{\K}f\,\D\pi=0\}$, the
set of centered $L^{2}$ functions. The mean and variance of $f\in L^{2}(\pi)$
are denoted as $\pi f:=\E_{\pi}f=\int_{\K}f\,\D\pi$ and $\var_{\pi}f=\norm{f-\pi f}^{2}_{2}$.

A function $f\in L^{2}(\K)$ has \emph{weak derivative} $\partial_{j}f\in L^{2}(\K)$
if $\int_{\K}f\,\partial_{j}\varphi\,\D x=-\int_{\K}(\partial_{j}f)\,\varphi\,\D x$
for every $\varphi\in C^{\infty}_{c}(\K)$. The \emph{Sobolev space}
$H^{1}(\K)$ consists of the $L^{2}$ functions whose weak derivatives
belong to $L^{2}$. Recall that it is a Hilbert space with $\inner{f,g}_{H^{1}}=\inner{f,g}_{L^{2}}+\inner{\nabla f,\nabla g}_{L^{2}}$.
We denote the closure of compactly-supported smooth functions with
respect to the $H^{1}$-norm by $H^{1}_{0}(\K):=\overline{C^{\infty}_{c}(\K)}^{H^{1}(\K)}$.
For a vector field $u=(u_{1},\ldots,u_{n})$, we use $H^{1}(\K;\R^{n}):=H^{1}(\K)^{n}$
and $\norm{\nabla u}^{2}_{2}=\sum^{n}_{i,j=1}\norm{\partial_{j}u_{i}}^{2}_{2}$. 

For $u\in H^{1}(\K;\R^{n})$, its \emph{weak divergence} is $\Div u:=\sum^{n}_{i=1}\partial_{i}u_{i}\in L^{2}(\K)$.
Equivalently, $\int_{\K}(\Div u)\,\varphi\,\D x=-\int_{\K}u\cdot\nabla\varphi\,\D x$
for any $\varphi\in C^{\infty}_{c}(\K)$. If $u\in H^{1}_{0}(\K;\R^{n})$,
approximation by compactly supported smooth vector fields gives $\int_{\K}\Div u\,\D\pi=0$. 

\paragraph{Functional inequalities.}

For $\delta(x):=\dist(x,\partial\K)=\min_{y\in\de\K}\abs{x-y}$ and
a probability measure $\pi$ over $\K$, its Poincar\'e constant
and improved Poincar\'e constant are defined as the smallest constants
satisfying that for any test function $f\in H^{1}(\K)$,
\begin{align}
\var_{\pi}f & \leq\cpi(\pi)\int_{\K}\abs{\nabla f}^{2}\,\D\pi\,,\tag{\ensuremath{\msf{PI}}}\label{eq:cpi-definition}\\
\var_{\pi}f & \leq\cipi(\pi)\int_{\K}\delta^{2}\,\abs{\nabla f}^{2}\,\D\pi\,.\tag{\ensuremath{\msf{IPI}}}\label{eq:ipi-definition}
\end{align}

\paragraph{Dirichlet form and spectral gap.}

Let $\mc G$ be a sub-$\sigma$-algebra of the Borel $\sigma$-algebra
on $\K$. Recall that the conditional expectation $P_{\mc G}f:=\E_{\pi}[f\mid\mc G]$
is the orthogonal projection of $L^{2}(\pi)$ onto the closed subspace
of $\mc G$-measurable functions:
\begin{equation}
P^{*}_{\mc G}=P_{\mc G}\,,\qquad P^{2}_{\mc G}=P_{\mc G}\,,\qquad\inner{f,(\Id-P_{\mc G})f}_{\pi}=\norm{f-P_{\mc G}f}^{2}_{2}\,.\label{eq:conditional-projection}
\end{equation}

For a $\pi$-reversible Markov operator $P$, its Dirichlet form is
$\mc E_{P}(f,f):=\inner{f,(\Id-P)f}_{\pi}$, and the spectral gap
of $P$ is defined as 
\[
\lda_{P}:=\inf_{0\neq f\in L^{2}_{0}(\pi)}\frac{\mc E_{P}(f,f)}{\norm f^{2}_{2}}\,.
\]
A positive spectral gap of positive semidefinite $P$ implies contraction
in $\chi^{2}$-divergence: if $q_{0}=\frac{\D\mu_{0}}{\D\pi}$ belongs
to $L^{2}(\pi)$ and $\mu_{N}=\mu_{0}P^{N}$, 
\begin{equation}
\chi^{2}(\mu_{N}\mmid\pi)=\norm{P^{N}(q_{0}-1)}^{2}_{2}\leq(1-\lda_{P})^{2N}\,\chi^{2}(\mu_{0}\mmid\pi)\leq\exp(-2\,\lda_{P}N)\,\chi^{2}(\mu_{0}\mmid\pi)\,.\label{eq:spectral-gap-to-chi2}
\end{equation}

\subsubsection{Proof ideas\label{subsec:Digesting-proof-ideas}}

We will define a bounded linear operator $T:F\to G$ for two Hilbert
spaces $F$ and $G$ so that $\norm f^{2}_{F}$ and $\norm{Tf}^{2}_{G}$
correspond to $\norm f^{2}_{2}$ and $\mc E_{P}(f,f)$ in our setting,
respectively. Then, our goal is to establish $\norm{Tf}^{2}_{G}\geq\lda\norm f^{2}_{F}$
for some $\lda>0$. A natural \emph{dual} approach is to find, for
each $f\in F$, a dual-certificate $g_{f}\in G$ such that $T^{*}g_{f}=f$
for the adjoint $T^{*}:G\to F$. If the certificate also satisfies
$\norm{g_{f}}^{2}_{G}\leq C\,\norm f^{2}_{F}$, then
\[
\norm f^{2}_{F}=\inner{f,f}_{F}=\inner{f,T^{*}g_{f}}_{F}=\inner{Tf,g_{f}}_{G}\leq\norm{Tf}_{G}\,\norm{g_{f}}_{G}\leq\sqrt{C}\,\norm{Tf}_{G}\,\norm f_{F}\,,
\]
which implies $\norm f^{2}_{F}\leq C\,\norm{Tf}^{2}_{G}$, and it
leads to $\lda\geq C^{-1}$. In summary, we should tackle two concrete
problems: $(i)$ find the dual certificate $g_{f}$ such that $T^{*}g_{f}=f$,
and $(ii)$ show $\norm{g_{f}}^{2}_{G}\leq C\,\norm f^{2}_{F}$ for
some $C>0$.

We now illustrate the main idea through Coordinate Hit-and-Run, whose
finite product structure makes the adjoint transparent. Fix an orthonormal
basis $e_{1},\ldots,e_{n}$, and let $P_{i}$ be conditional expectation
given all coordinates except the $i$-th one; equivalently, $P_{i}f$
is the uniform average of $f$ on each $e_{i}$-parallel chord. By
\eqref{eq:conditional-projection}, 
\begin{align*}
P_{\CHAR} & =\frac{1}{n}\sum^{n}_{i=1}P_{i}\,,\\
\mc E_{\CHAR}(f,f) & :=\inner{f,(\Id-P_{\CHAR})f}_{\pi}=\frac{1}{n}\sum^{n}_{i=1}\norm{f-P_{i}f}^{2}_{2}\,.
\end{align*}

In this case, we can take $F:=L^{2}_{0}(\pi)$ and $G:=\{g=(g_{1},\dots,g_{n}):g_{i}\in L^{2}(\pi)\}$
with norm $\norm f^{2}_{F}=\norm f^{2}_{2}$ and $\norm g^{2}_{G}=\frac{1}{n}\sum^{n}_{i=1}\norm{g_{i}}^{2}_{2}$.
Also, we define $T:F\to G$ by $(Tf)_{i}=(\Id-P_{i})f$ for $i\in[n]$,
so $\norm{Tf}^{2}_{G}=\mc E_{\CHAR}(f,f)$ as desired. Let us now
compute the adjoint $T^{*}$. For $g=(g_{1},\dots,g_{n})\in G$, using
$P_{i}=P^{*}_{i}$,
\[
\inner{Tf,g}_{G}=\frac{1}{n}\sum_{i}\inner{(\Id-P_{i})f,g_{i}}_{\pi}=\frac{1}{n}\sum_{i}\inner{f,(\Id-P_{i})g_{i}}_{\pi}\,.
\]
Hence, $T^{*}g=\frac{1}{n}\sum_{i}(\Id-P_{i})g_{i}$. 

With the dual-certificate approach in mind, we address the first problem
of finding $g$ such that $T^{*}g=f$. To this end, we could consider
two sufficient conditions given as 
\[
P_{i}g_{i}=0\quad\text{for }i\in[n]\qquad\&\qquad\frac{1}{n}\sum^{n}_{i=1}g_{i}=f\,.
\]
A convenient way to enforce the first condition is to set $g_{i}=n\,\de_{i}u_{i}$
for each $i\in[n]$ for some $u\in H^{1}_{0}(\K;\R^{n})$ extended
by zero outside of $\K$. Indeed, for a.e. $z\in e^{\perp}_{i}$,
if $I_{z}:=\{s\in\R:z+se_{i}\in\K\}$, then for $t\in I_{z}$, the
fundamental theorem of calculus leads to
\[
(P_{i}g_{i})(z+te_{i})=\frac{n}{\abs{I_{z}}}\int_{I_{z}}\partial_{i}u_{i}(z+se_{i})\,\D s=0\,.
\]
Then, the second condition corresponds to $\Div u=f$. In summary,
the first problem comes down to finding $u\in H^{1}_{0}(\K;\R^{n})$
such that $\Div u=f$.

We now address the second problem of showing $\norm g^{2}_{G}\leq C\,\norm f^{2}_{F}$
for some $C$. Since 
\[
\norm g^{2}_{G}=\frac{1}{n}\sum^{n}_{i=1}\norm{g_{i}}^{2}_{2}=n\sum^{n}_{i=1}\norm{\de_{i}u_{i}}^{2}_{2}\leq n\,\norm{\nabla u}^{2}_{2}\,,
\]
the second problem can be reduced to the following optimization problem:
among functions $u$ with $\Div u=f$, find a function $u$ that minimizes
$\norm{\nabla u}^{2}_{2}$. Namely, one should minimize the Frobenius
norm $\norm{\grad u}_{2}$ with trace constraint $\Div u=\tr(\grad u)=f$.
This is where the Babu\v{s}ka--Aziz constant comes in:
\[
C_{\msf{BA}}(\pi):=\sup_{0\neq f\in L^{2}_{0}(\pi)}\,\inf_{\substack{u\in H^{1}_{0}(\K;\R^{n})\\
\Div u=f
}
}\frac{\norm{\nabla u}^{2}_{2}}{\norm f^{2}_{2}}\,.
\]
The minimizer satisfies $\Div u=f$ and $\norm{\nabla u}^{2}_{2}\leq C_{\msf{BA}}(\pi)\,\norm f^{2}_{2}$.
Consequently, $\norm g^{2}_{G}\leq nC_{\msf{BA}}(\pi)\,\norm f^{2}_{F}$,
and the dual argument gives $\lda_{\CHAR}\geq\frac{1}{nC_{\msf{BA}}(\pi)}$.
We formalize this functional analytic result as follows.
\begin{lem}
\label{lem:known-divergence} Suppose $C_{\msf{BA}}(\pi)<\infty$.
Then, for every $f\in L^{2}_{0}(\pi)$, there exists $u\in H^{1}_{0}(\K;\R^{n})$
such that 
\begin{equation}
\Div u=f\,,\qquad\mbox{ and }\qquad\norm{\nabla u}^{2}_{2}\leq C_{\msf{BA}}(\pi)\,\norm f^{2}_{2}\,.\label{eq:div-solver-known}
\end{equation}
\end{lem}

\begin{proof}
Finiteness of $C_{\mathrm{BA}}(\pi)$ makes the feasible set in the
inner infimum nonempty. Equip $H^{1}_{0}(\K;\Rn)$ with the norm $\|\nabla u\|_{2}$.
By the Poincar\'e inequality for $H^{1}_{0}(\K;\Rn)$, this is an
equivalent Hilbert norm on $H^{1}_{0}(\K;\Rn)$. Since $\norm{\Div u}_{2}\leq\sqrt{n}\,\norm{\nabla u}_{2}$,
it follows that $\Div:H^{1}_{0}(\K;\Rn)\to L^{2}_{0}(\pi)$ is a bounded
linear operator (so continuous). Hence, $\Div^{-1}(\{f\})=\{u\in H^{1}_{0}(\K;\R^{n}):\Div u=f\}$
is a nonempty closed convex subset of $H^{1}_{0}(\K;\R^{n})$. By
the Hilbert projection theorem, the infimum in the definition of $C_{\mathrm{BA}}(\pi)$
is attained, and its minimizer satisfies \eqref{eq:div-solver-known}.
\end{proof}

In \S\ref{sec:Spectral_gap_to_BA}, we formally bound the spectral
gaps of HAR and CHAR in terms of the Babu\v{s}ka--Aziz constant
of the target distribution. Notably, the spectral-gap bound for CHAR
is smaller by a factor of the dimension. The rest of the analysis
is independent of the Markov chain and will apply to both HAR and
CHAR. The next step is the following result from \cite{Zsuppan2020}
which relates the Babu\v{s}ka--Aziz constant to the Improved Poincar\'e
constant. We give a self-contained proof of this in \S\ref{sec:BA_to_Poincare}. 
\begin{thm}
\label{thm:BA_to_IPI} For a bounded convex domain in $\Rn$, we have
$C_{\msf{BA}}\le1+4\cipi$.
\end{thm}

Then, in \S\ref{sec:bounding_IPI}, we show that $\cipi(\pi)\lesssim n^{2}\cpi(\pi)$
for a convex body containing a unit ball (Theorem~\ref{thm:l1-ipi}
and Corollary~\ref{cor:ipi-pi}). These ingredients suffice to prove
the main results. 

\subsection{Proof of main result}
\begin{proof}[Proof of Theorem~\ref{thm:main}]
 Note that $\cpi(\pi)\ge\norm{\cov\pi}_{\op}\gtrsim1/n$ when $B(0,1)\subset\K$
(see \cite[Lemma~2.1]{KV26unified}). By applying Propositions~\ref{prop:HR-comparison}
and~\ref{prop:CHAR-comparison} with Theorem~\ref{thm:BA_to_IPI}
($C_{\msf{BA}}\leq1+4\cipi$) and Corollary~\ref{cor:ipi-pi} ($\cipi\lesssim n^{2}\cpi/r^{2}$
with $r=1$),
\[
\lda_{\HR}\gtrsim\frac{1}{n^{2}\cpi(\pi)}\,,\qquad\lda_{\CHAR}\gtrsim\frac{1}{n^{3}\cpi(\pi)}\,.
\]
Using $\cpi(\pi)\lesssim\norm{\cov\pi}_{\op}\log n$ \cite{Klartag23log},
we complete the proof.
\end{proof}

\begin{proof}
[Proof of Corollary~\ref{cor:mixing-har-char}] This follows from
the spectral-gap results in Theorem~\ref{thm:main} and \eqref{eq:spectral-gap-to-chi2}.
\end{proof}

\section{Spectral gap via Babu\v{s}ka--Aziz constant \label{sec:Spectral_gap_to_BA}}

\subsection{Hit-and-Run}

Let $\sigma$ be uniform probability measure on $\Sph^{n-1}$. For
$\theta\in\Sph^{n-1}$, let $P_{\theta}$ be conditional expectation
with respect to the orthogonal projection onto $\theta^{\perp}$,
so $P_{\theta}f$ is the uniform average of $f$ on each $\theta$-parallel
chord. Then,
\begin{align*}
P_{\HR} & =\int_{\Sph^{n-1}}P_{\theta}\,\D\sigma(\theta)\,,\\
\mc E_{\HR}(f,f) & =\inner{f,(\Id-P_{\HR})f}_{\pi}=\int_{\Sph^{n-1}}\norm{f-P_{\theta}f}^{2}_{2}\,\D\sigma(\theta)\,.
\end{align*}
For $\theta$ drawn from $\Sph^{n-1}$, the natural dual certificate
is $g_{\theta}(x):=n\,\theta^{\T}\nabla u(x)\theta$.
\begin{prop}
\label{prop:HR-comparison} If $C_{\msf{BA}}(\pi)<\infty$, then 
\[
\lda_{\HR}\geq\frac{n+2}{n\,(2+C_{\msf{BA}}(\pi))}\gtrsim\frac{1}{1+C_{\msf{BA}}(\pi)}\,.
\]
\end{prop}

\begin{proof}
Fix $f\in L^{2}_{0}(\pi)$, and pick $u\in H^{1}_{0}(\K;\Rn)$ given
by Lemma~\ref{lem:known-divergence} that satisfies $\Div u=f$ and
$\norm{\nabla u}^{2}_{2}\leq C_{\msf{BA}}(\pi)\,\norm f^{2}_{2}$.
For $\theta\in\Sph^{n-1}$, define 
\[
g_{\theta}(x):=n\,\theta^{\T}\nabla u(x)\,\theta\qquad\text{for a.e. }x\in\K\,.
\]
 Let $u\in H^{1}_{0}(\K;\R^{n})$ and choose $u_{k}\in C^{\infty}_{c}(\K;\R^{n})$
converging to $u$ in $H^{1}$. Define $g_{\theta,k}:=n\,\theta^{\T}\nabla u_{k}\theta$.
The fundamental theorem of calculus on every $\theta$-parallel line
gives $P_{\theta}g_{\theta,k}=0$, because $u_{k}$ is compactly supported.
Since $g_{\theta,k}\to g_{\theta}$ in $L^{2}(\pi)$ and conditional
expectation is an $L^{2}$ contraction, $P_{\theta}g_{\theta}=\lim_{k\to\infty}P_{\theta}g_{\theta,k}=0$
in $L^{2}$. Hence, $0=\inner{f,P_{\theta}g_{\theta}}_{\pi}=\inner{P_{\theta}f,g_{\theta}}_{\pi}$.

Recall that for a square matrix $M$, rotational invariance of $\D\sigma(\theta)$
gives 
\begin{align}
\E_{\theta}[\theta\theta^{\T}] & =\frac{1}{n}\,\Id\,,\label{eq:HR-second-moment}\\
\E_{\theta}[(\theta^{\T}M\theta)^{2}] & =\frac{(\tr M)^{2}+2\,\norm{\sym M}^{2}_{\frob}}{n\,(n+2)}\,,\label{eq:HR-fourth-moment}
\end{align}
where the second identity follows from $\E[\theta_{i}\theta_{j}\theta_{k}\theta_{\ell}]=\frac{\delta_{ij}\delta_{k\ell}+\delta_{ik}\delta_{j\ell}+\delta_{i\ell}\delta_{jk}}{n\,(n+2)}$.
By \eqref{eq:HR-second-moment} and $\Div u=f$, 
\[
\int_{\Sph^{n-1}}g_{\theta}\,\D\sigma(\theta)=\tr\nabla u=f\,.
\]
Therefore, using $\inner{P_{\theta}f,g_{\theta}}_{\pi}=0$ and Cauchy--Schwarz,
\begin{align}
\norm f^{2}_{2} & =\int_{\Sph^{n-1}}\inner{f,g_{\theta}}_{\pi}\,\D\sigma(\theta)=\int_{\Sph^{n-1}}\inner{f-P_{\theta}f,g_{\theta}}_{\pi}\,\D\sigma(\theta)\,,\nonumber \\
\norm f^{4}_{2} & \leq\mc E_{\HR}(f,f)\,\int_{\Sph^{n-1}}\norm{g_{\theta}}^{2}_{2}\,\D\sigma(\theta)\,.\label{eq:duality-CS-HR}
\end{align}
For a square matrix $M$, we denote its symmetrization by $\sym M:=(M+M^{\T})/2$.
By integration by parts, for $u\in H^{1}_{0}(\K;\R^{n})$, 
\begin{equation}
2\,\norm{\sym\nabla u}^{2}_{2}=\norm{\nabla u}^{2}_{2}+\norm{\Div u}^{2}_{2}\,,\label{eq:prelim-sym-identity}
\end{equation}
since the cross term satisfies $\sum_{i,j}\int_{\K}(\partial_{j}u_{i})(\partial_{i}u_{j})\,\D\pi=\int_{\K}(\Div u)^{2}\,\D\pi$.
Applying \eqref{eq:HR-fourth-moment} to $M=\nabla u$, and then using
\eqref{eq:prelim-sym-identity}, $\Div u=f$, and $\norm{\nabla u}^{2}_{2}\leq C_{\msf{BA}}(\pi)\,\norm f^{2}_{2}$,
we obtain
\begin{align*}
\int_{\Sph^{n-1}}\norm{g_{\theta}}^{2}_{2}\,\D\sigma(\theta) & =\frac{n}{n+2}\,\E_{\pi}[(\Div u)^{2}+2\,\norm{\sym\nabla u}^{2}_{\frob}]=\frac{n}{n+2}\,(\norm f^{2}_{2}+2\,\norm{\sym\nabla u}^{2}_{2})\\
 & \leq\frac{n\,\bpar{2+C_{\msf{BA}}(\pi)}}{n+2}\,\norm f^{2}_{2}\,.
\end{align*}
Substitution into \eqref{eq:duality-CS-HR} gives 
\[
\mc E_{\HR}(f,f)\geq\frac{n+2}{n\,\bpar{2+C_{\msf{BA}}(\pi)}}\,\norm f^{2}_{2}\,.
\]
Taking the infimum over nonzero $f\in L^{2}_{0}(\pi)$ completes the
proof.
\end{proof}

\subsection{Coordinate Hit-and-Run}
\begin{prop}
\label{prop:CHAR-comparison} If $C_{\msf{BA}}(\pi)<\infty$, then
\[
\lda_{\CHAR}\geq\frac{1}{nC_{\msf{BA}}(\pi)}\,.
\]
\end{prop}

\begin{proof}
Fix $f\in L^{2}_{0}(\pi)$ and let $u$ be given by Lemma~\ref{lem:known-divergence}.
For $i=1,\ldots,n$, set $g_{i}:=n\,\de_{i}u_{i}$ and note that $\frac{1}{n}\sum^{n}_{i=1}g_{i}=\Div u=f$.
A similar argument in the proof of Proposition~\ref{prop:HR-comparison}
gives $P_{i}g_{i}=0$, and thus $0=\inner{f,P_{i}g_{i}}_{\pi}=\inner{P_{i}f,g_{i}}_{\pi}$.
Then, 
\begin{align}
\norm f^{2}_{2} & =\frac{1}{n}\sum^{n}_{i=1}\inner{f,g_{i}}_{\pi}=\frac{1}{n}\sum^{n}_{i=1}\inner{f-P_{i}f,g_{i}}_{\pi}\,,\nonumber \\
\norm f^{4}_{2} & \leq\mc E_{\CHAR}(f,f)\,\frac{1}{n}\sum^{n}_{i=1}\norm{g_{i}}^{2}_{2}=n\,\mc E_{\CHAR}(f,f)\,\sum^{n}_{i=1}\norm{\partial_{i}u_{i}}^{2}_{2}\,.\label{eq:duality-CS-CHAR}
\end{align}
By Lemma~\ref{lem:known-divergence}, 
\[
\sum^{n}_{i=1}\norm{\partial_{i}u_{i}}^{2}_{2}\leq\norm{\nabla u}^{2}_{2}\leq C_{\msf{BA}}(\pi)\,\norm f^{2}_{2}\,.
\]
Substitution into \eqref{eq:duality-CS-CHAR} gives $\mc E_{\CHAR}(f,f)\geq\frac{1}{nC_{\msf{BA}}(\pi)}\,\norm f^{2}_{2}$,
and taking the infimum over nonzero $f\in L^{2}_{0}(\pi)$ finishes
the proof.
\end{proof}

\section{Bounding the improved Poincar\'e constant \label{sec:bounding_IPI}}

Since $C_{\msf{BA}}\le1+4\cipi$ for a bounded convex domain (Theorem~\ref{thm:BA_to_IPI}),
we are only left with bounding the improved Poincar\'e constant \eqref{eq:ipi-definition}
in terms of the usual Poincar\'e constant. Before we proceed, we
briefly survey prior work on \eqref{eq:ipi-definition}. Improved
Poincar\'e inequalities with the distance to the boundary as a weight
have been studied for John domains (which include convex domains),
starting with \cite{HurriSyrjanen1994} and further developed in \cite{DD08improved,CW10poincare}.
In particular, for convex domains, \cite{CW10poincare} gives a quantitative
bound in terms of the \emph{eccentricity} of the domain. In our notation,
this implies $\cipi(\pi)\le C(n)\,R^{2}/r^{2}$, where $R$ and $r$
denote the outer-radius and in-radius of the domain, respectively,
and the dimension dependence is absorbed into $C(n)$. For our application,
we need a bound in terms of the ordinary Poincar\'e constant with
explicit dimension dependence, so in Corollary~\ref{cor:ipi-pi}
we show $\cipi(\pi)\lesssim n^{2}\cpi(\pi)/r^{2}$.
\begin{thm}
\label{thm:l1-ipi} For any locally Lipschitz function $f:\Rn\to\R$
and a convex body $\K\subset\Rn$ containing a ball of radius $r$,
we have
\[
\inf_{z}\int_{\K}\abs{f(x)-z}\,\D x\lesssim\frac{n\sqrt{\cpi(\pi)}}{r}\int_{\K}\delta(x)\,\abs{\nabla f}\,\D x\,.
\]
\end{thm}

This gives the following immediate corollary:
\begin{cor}
\label{cor:ipi-pi} For a convex body $\K$ in $\Rn$ whose maximum
inscribed ball has radius $r$, we have
\[
\cipi(\pi)\lesssim\frac{n^{2}\cpi(\pi)}{r^{2}}\,.
\]
\end{cor}

\begin{proof}
Since $C^{\infty}(\Rn)|_{\K}$ is dense in $H^{1}(\K)$, it suffices
to prove the claim for $f\in C^{\infty}(\mathbb{R}^{n})|_{\K}$. Then,
Theorem~\ref{thm:l1-ipi} holds with $z$ set to the median of $f$,
i.e., for any $g:\Rn\rightarrow\R$ with median $0$, we have
\[
\int_{\K}\abs{g(x)}\,\D x\lesssim\frac{n\sqrt{\cpi(\pi)}}{r}\int_{\K}\delta(x)\,\abs{\nabla g}\,\D x\,.
\]
Now set $u(x):=(f(x)-z_{f})^{+}$ where $z_{f}$ is the median of
$f$ and apply the above to $g=u^{2}$. Then, $0$ is the median of
$g$. Denoting $C:=n\cpi(\pi)^{1/2}/r$, we have
\[
\int_{\K}u^{2}\,\D x\lesssim C\int_{\K}\delta(x)\,u\,\abs{\nabla u}\,\D x\leq C\,\Bpar{\int_{\K}u^{2}\,\D x}^{1/2}\Bpar{\int_{\K}\delta(x)^{2}\,\abs{\nabla u}^{2}\,\D x}^{1/2}\,.
\]
 Hence, $\int_{\K}u^{2}\,\D x\lesssim C^{2}\int_{\K}\delta(x)^{2}\,\abs{\nabla u}^{2}\,\D x$.

We get the same conclusion with $v=(z_{f}-f(x))^{+}$ in place of
$u$. Adding these,
\[
\var_{\pi}f\leq\int_{\K}(f-z_{f})^{2}\,\D\pi\lesssim\frac{n^{2}\cpi(\pi)}{r^{2}}\int_{\K}\delta(x)^{2}\,\abs{\nabla f}^{2}\,\D\pi\,,
\]
which completes the proof.
\end{proof}

We now give an elementary proof of the theorem.
\begin{proof}
[Proof of Theorem~\ref{thm:l1-ipi}] We may assume that $B(0,r)\subset\K$
by translation. The idea is just that since $\K$ contains a ball
of radius $r$, we will consider a slight contraction of $\K$ by
roughly $(1-1/n)$. Set $\alpha=1-\frac{1}{2n}$, and note that 
\[
\inf_{z}\int_{\K}\abs{f(x)-z}\,\D x\le\int_{\K}\abs{f(x)-f(\alpha x)}\,\D x+\inf_{z}\int_{\K}\abs{f(\alpha x)-z}\,\D x\,.
\]
For the first term, since the diameter of isotropic convex bodies
is at most $\sqrt{n(n+2)}\leq n+1$ \cite[Theorem~4.1]{KLS95isop},
we have $\sup_{\K}\abs x\lesssim n\,\norm{\cov\pi}^{1/2}_{\op}\leq n\cpi(\pi)^{1/2}$.
Hence, for the gauge function $\norm x_{\K}:=\inf\{t>0:x\in t\K\}$,
\begin{align*}
\int_{\K}\abs{f(x)-f(\alpha x)}\,\D x & \le\int_{\K}\int^{1}_{\alpha}\abs x\,\abs{\nabla f(tx)}\,\D t\D x\lesssim n\sqrt{\cpi}\int^{1}_{\alpha}t^{-n}\int_{t\K}\abs{\nabla f(y)}\,\D y\D t\\
 & =n\sqrt{\cpi}\int_{\K}\abs{\nabla f(y)}\int^{1}_{\max\{\alpha,\norm y_{\K}\}}t^{-n}\,\D t\D y\\
 & \leq n\sqrt{\cpi}\int_{\K}\abs{\nabla f(y)}\,\alpha^{-n}\,\int^{1}_{\max\{\alpha,\norm y_{\K}\}}\D t\D y\\
 & \le n\sqrt{\cpi}\int_{\K}\abs{\nabla f(y)}\cdot2\,(1-\norm y_{\K})\,\D y\\
 & \lesssim\frac{n\sqrt{\cpi}}{r}\int_{\K}\delta(y)\,\abs{\nabla f(y)}\,\D y\,,
\end{align*}
where in the last line, we used $\delta(y)\ge r\,(1-\norm y_{\K})$,
since $\norm y_{\K}\,\K+(1-\norm y_{\K})\,B(0,r)\subset\K$.

For the second term, using the standard equivalence between the $L^{1}$
and $L^{2}$-Poincar\'e constants for logconcave measures, and noting
that the constant is $\alpha^{2}\cpi$ for $\alpha\K$,
\[
\inf_{z}\int_{\K}\abs{f(\alpha x)-z}\,\D x=\alpha^{-n}\inf_{z}\int_{\alpha\K}\abs{f(x)-z}\,\D x\le\alpha^{-n}\alpha\sqrt{\cpi}\int_{\alpha\K}\abs{\nabla f}\,\D x\lesssim\frac{n\sqrt{\cpi}}{r}\int_{\alpha\K}\delta(x)\,\abs{\nabla f}\,\D x\,,
\]
where we used $\delta(x)\ge\frac{r}{2n}$ for $x\in\alpha\K$. The
theorem follows by adding these two bounds.
\end{proof}

\begin{acknowledgement*}
This work was supported in part by NSF Awards CCF-2504995, CCF-2236669,
CCF-2504994, and a Simons Investigator award. The second author is
grateful to Laci Lov\'asz for many inspiring discussions on the problem. 
\end{acknowledgement*}
\bibliographystyle{alpha}
\bibliography{main}

@article{HP83inequalities,
	author = {Cornelius O. Horgan and Lawrence E. Payne},
	journal = {Archive for Rational Mechanics and Analysis},
	number = {2},
	pages = {165--179},
	title = {On inequalities of {K}orn, {F}riedrichs and {B}abu\v{s}ka-{A}ziz},
	volume = {82},
	year = {1983}}

@incollection{BA72survey,
  author    = {Ivo Babu{\v{s}}ka and Abdul K. Aziz},
  title     = {Survey lectures on the mathematical foundations of the finite element method},
  booktitle = {The Mathematical Foundations of the Finite Element Method with Applications to Partial Differential Equations
               ({Proc. Sympos., Univ. Maryland, Baltimore, Md., 1972})},
  pages     = {1--359},
  publisher = {Academic Press},
  address   = {New York},
  year      = {1972},
  note      = {With the collaboration of G. Fix and R. B. Kellogg},
  mrnumber  = {0421106}
}

@article{CE25localization,
	author = {Yuansi Chen and Ronen Eldan},
	journal = {Duke Mathematical Journal},
	number = {8},
	pages = {1431--1510},
	title = {Localization schemes: a framework for proving mixing bounds for {M}arkov chains},
	volume = {174},
	year = {2025}}

@article{KV26unified,
	author = {Kook, Yunbum and Vempala, Santosh S.},
	journal = {arXiv preprint arXiv:2606.12694},
	title = {A unified complexity bound for logconcave sampling},
	year = {2026}}

@article{KVZ26INO,
	author = {Kook, Yunbum and Vempala, Santosh S. and Zhang, Matthew S.},
	doi = {https://doi.org/10.1002/rsa.70061},
	eprint = {https://onlinelibrary.wiley.com/doi/pdf/10.1002/rsa.70061},
	journal = {Random Structures \& Algorithms},
	number = {3},
	pages = {e70061},
	title = {In-and-{O}ut: algorithmic diffusion for sampling convex bodies},
	url = {https://onlinelibrary.wiley.com/doi/abs/10.1002/rsa.70061},
	volume = {68},
	year = {2026}}

@article{KV25localization,
	author = {Kook, Yunbum and Vempala, Santosh S.},
	journal = {arXiv preprint arXiv:2512.10848},
	title = {The localization method for high-dimensional inequalities},
	year = {2025}}

@article{CE25hitandrun,
	author = {Chen, Yuansi and Eldan, Ronen},
	journal = {Discrete \& Computational Geometry},
	number = {3},
	pages = {747--794},
	title = {Hit-and-{R}un mixing via localization schemes},
	volume = {75},
	year = {2026}}

@article{KV26zeroLC,
	author = {Kook, Yunbum and Vempala, Santosh S.},
	journal = {arXiv preprint arXiv:2507.18021},
	title = {Zeroth-order logconcave sampling},
	year = {2025}}

@inproceedings{KV25sampling,
	author = {Kook, Yunbum and Vempala, Santosh S.},
	booktitle = {{S}ymposium on {T}heory of {C}omputing},
	doi = {10.1145/3717823.3718202},
	isbn = {979-8-4007-1510-5},
	mrclass = {68Q87},
	mrnumber = {4928485},
	pages = {924--932},
	publisher = {ACM},
	title = {Sampling and integration of logconcave functions by algorithmic diffusion},
	url = {https://doi.org/10.1145/3717823.3718202},
	year = {2025}}

@incollection{vempala05geometric,
	author = {Vempala, Santosh S.},
	booktitle = {Combinatorial and Computational Geometry},
	isbn = {978-0-521-84862-6; 0-521-84862-8},
	mrclass = {90C59 (52B55 60D05 60G50)},
	mrnumber = {2178341},
	pages = {577--616},
	publisher = {Cambridge Univ. Press},
	series = {Math. Sci. Res. Inst. Publ.},
	title = {Geometric random walks: a survey},
	volume = {52},
	year = {2005}}

@article{NS22char,
  title={On the mixing time of coordinate hit-and-run},
  author={Narayanan, Hariharan and Srivastava, Piyush},
  journal={Combinatorics, Probability and Computing},
  volume={31},
  number={2},
  pages={320--332},
  year={2022},
  publisher={Cambridge University Press}
}

@article{NRS25sampling,
	author = {Narayanan, Hariharan and Rajaraman, Amit and Srivastava, Piyush},
	doi = {10.1007/s00440-024-01341-w},
	fjournal = {Probability Theory and Related Fields},
	issn = {0178-8051,1432-2064},
	journal = {Probability Theory and Related Fields},
	mrclass = {60J05 (52A20 68W40)},
	mrnumber = {4898100},
	number = {3-4},
	pages = {1169--1232},
	title = {Sampling from convex sets with a cold start using multiscale decompositions},
	url = {https://doi.org/10.1007/s00440-024-01341-w},
	volume = {191},
	year = {2025}}

@article{LV23char,
	author = {Laddha, Aditi and Vempala, Santosh S.},
	doi = {10.1007/s00454-023-00497-x},
	fjournal = {Discrete \& Computational Geometry. An International Journal of Mathematics and Computer Science},
	issn = {0179-5376,1432-0444},
	journal = {Discrete \& Computational Geometry},
	mrclass = {60J05 (37A25 37A30)},
	mrnumber = {4627347},
	mrreviewer = {Malin\ P.\ Forsstr\"om},
	number = {2},
	pages = {406--425},
	title = {Convergence of {G}ibbs sampling: coordinate hit-and-run mixes fast},
	url = {https://doi.org/10.1007/s00454-023-00497-x},
	volume = {70},
	year = {2023}}

@book{evans10partial,
	author = {Evans, Lawrence C.},
	doi = {10.1090/gsm/019},
	edition = {Second},
	isbn = {978-0-8218-4974-3},
	mrclass = {35-01},
	mrnumber = {2597943},
	mrreviewer = {Diego M. Maldonado},
	pages = {xxii+749},
	publisher = {American Mathematical Society},
	series = {Graduate Studies in Mathematics},
	title = {Partial differential equations},
	url = {https://doi.org/10.1090/gsm/019},
	volume = {19},
	year = {2010}}

@article{lovasz90compute,
	author = {Lov\'{a}sz, L\'{a}szl\'{o}},
	journal = {Jber. d. Dt. Math.-Verein, Jubil{\"a}umstagung},
	pages = {138--151},
	publisher = {DIMACS, Center for Discrete Mathematics and Theoretical Computer Science},
	title = {How to compute the volume?},
	year = {1990}}

@article{Lovasz99hit,
	author = {Lov\'{a}sz, L\'{a}szl\'{o}},
	doi = {10.1007/s101070050099},
	id = {Lov{\'a}sz1999},
	isbn = {1436-4646},
	journal = {Mathematical Programming},
	number = {3},
	pages = {443--461},
	title = {Hit-and-run mixes fast},
	url = {https://doi.org/10.1007/s101070050099},
	volume = {86},
	year = {1999}}

@article{Smith84HAR,
	author = {Smith, Robert L.},
	doi = {10.1287/opre.32.6.1296},
	fjournal = {Operations Research},
	issn = {0030-364X},
	journal = {Operations Research},
	mrclass = {90B15 (65C05)},
	mrnumber = {775260},
	number = {6},
	pages = {1296--1308},
	title = {Efficient {M}onte {C}arlo procedures for generating points uniformly distributed over bounded regions},
	url = {https://doi.org/10.1287/opre.32.6.1296},
	volume = {32},
	year = {1984}}

@article{LV06hit,
	author = {Lov\'{a}sz, L\'{a}szl\'{o} and Vempala, Santosh S.},
	doi = {10.1137/s009753970544727x},
	issn = {1095-7111},
	journal = {SIAM Journal on Computing},
	number = {4},
	pages = {985--1005},
	publisher = {Society for Industrial and Applied Mathematics (SIAM)},
	title = {Hit-and-Run from a corner},
	url = {http://dx.doi.org/10.1137/S009753970544727X},
	volume = {35},
	year = {2006}}

@article{LV07geometry,
	author = {Lov\'{a}sz, L\'{a}szl\'{o} and Vempala, Santosh S.},
	doi = {10.1002/rsa.20135},
	fjournal = {Random Structures \& Algorithms},
	issn = {1042-9832},
	journal = {Random Structures \& Algorithms},
	mrclass = {94A20 (60C05 60J10 60J22 62D05 65C10)},
	mrnumber = {2309621},
	mrreviewer = {Hsien-Kuei Hwang},
	number = {3},
	pages = {307--358},
	title = {The geometry of logconcave functions and sampling algorithms},
	url = {https://doi.org/10.1002/rsa.20135},
	volume = {30},
	year = {2007}}

@article{LV06simulated,
	author = {Lov\'{a}sz, L\'{a}szl\'{o} and Vempala, Santosh S.},
	doi = {10.1016/j.jcss.2005.08.004},
	fjournal = {Journal of Computer and System Sciences},
	issn = {0022-0000},
	journal = {Journal of Computer and System Sciences},
	mrclass = {68U05 (52A20 60C05 68W20)},
	mrnumber = {2205290},
	mrreviewer = {Carla Peri},
	number = {2},
	pages = {392--417},
	title = {Simulated annealing in convex bodies and an {$O^*(n^4)$} volume algorithm},
	url = {https://doi.org/10.1016/j.jcss.2005.08.004},
	volume = {72},
	year = {2006}}

@article{LS93random,
	author = {Lov\'{a}sz, L\'{a}szl\'{o} and Simonovits, Mikl\'{o}s},
	doi = {10.1002/rsa.3240040402},
	fjournal = {Random Structures \& Algorithms},
	issn = {1042-9832},
	journal = {Random Structures \& Algorithms},
	mrclass = {90C27 (52B55 90C15)},
	mrnumber = {1238906},
	mrreviewer = {Gerard Sierksma},
	number = {4},
	pages = {359--412},
	title = {Random walks in a convex body and an improved volume algorithm},
	url = {https://doi.org/10.1002/rsa.3240040402},
	volume = {4},
	year = {1993}}

@article{eldan13thin,
	author = {Eldan, Ronen},
	doi = {10.1007/s00039-013-0214-y},
	fjournal = {Geometric and Functional Analysis},
	issn = {1016-443X},
	journal = {Geometric and Functional Analysis},
	mrclass = {60D05 (52A40 60H10)},
	mrnumber = {3053755},
	mrreviewer = {Christian Rau},
	number = {2},
	pages = {532--569},
	title = {Thin shell implies spectral gap up to polylog via a stochastic localization scheme},
	url = {https://doi.org/10.1007/s00039-013-0214-y},
	volume = {23},
	year = {2013}}

@article{CV18Gaussian,
	author = {Cousins, Ben and Vempala, Santosh S.},
	doi = {10.1137/15M1054250},
	fjournal = {SIAM Journal on Computing},
	issn = {0097-5397},
	journal = {SIAM Journal on Computing},
	mrclass = {68W20 (52A38 60D05 65C40)},
	mrnumber = {3818340},
	number = {3},
	pages = {1237--1273},
	publisher = {Society for Industrial and Applied Mathematics (SIAM)},
	title = {Gaussian cooling and {$O^*(n^3)$} algorithms for volume and {G}aussian volume},
	url = {https://doi.org/10.1137/15M1054250},
	volume = {47},
	year = {2018}}

@article{KLS97random,
	author = {Kannan, Ravi and Lov\'{a}sz, L\'{a}szl\'{o} and Simonovits, Mikl\'{o}s},
	doi = {10.1002/(SICI)1098-2418(199708)11:1<1::AID-RSA1>3.0.CO;2-X},
	fjournal = {Random Structures \& Algorithms},
	issn = {1042-9832},
	journal = {Random Structures \& Algorithms},
	mrclass = {68Q25 (52A20 52A38 60J15 60J20)},
	mrnumber = {1608200},
	mrreviewer = {Mark R. Jerrum},
	number = {1},
	pages = {1--50},
	title = {Random walks and an {$O^*(n^5)$} volume algorithm for convex bodies},
	url = {https://doi.org/10.1002/(SICI)1098-2418(199708)11:1<1::AID-RSA1>3.0.CO;2-X},
	volume = {11},
	year = {1997}}

@article{KV06simulated,
	author = {Kalai, Adam T. and Vempala, Santosh S.},
	doi = {10.1287/moor.1060.0194},
	fjournal = {Mathematics of Operations Research},
	issn = {0364-765X},
	journal = {Mathematics of Operations Research},
	mrclass = {90C25 (90C59)},
	mrnumber = {2233996},
	number = {2},
	pages = {253--266},
	publisher = {Institute for Operations Research and the Management Sciences (INFORMS)},
	title = {Simulated annealing for convex optimization},
	url = {https://doi.org/10.1287/moor.1060.0194},
	volume = {31},
	year = {2006}}

@article{DFK91random,
	author = {Dyer, Martin and Frieze, Alan and Kannan, Ravi},
	doi = {10.1145/102782.102783},
	fjournal = {Journal of the ACM},
	issn = {0004-5411},
	journal = {Journal of the ACM},
	mrclass = {68U05 (52B55 68Q25)},
	mrnumber = {1095916},
	number = {1},
	pages = {1--17},
	title = {A random polynomial-time algorithm for approximating the volume of convex bodies},
	url = {https://doi.org/10.1145/102782.102783},
	volume = {38},
	year = {1991}}

@incollection{LS90mixing,
	author = {Lov\'{a}sz, L\'{a}szl\'{o} and Simonovits, Mikl\'{o}s},
	booktitle = {{S}ymposium on {F}oundations of {C}omputer {S}cience},
	doi = {10.1109/FSCS.1990.89553},
	mrclass = {68Q25 (52A38 52B55 60J05)},
	mrnumber = {1150706},
	mrreviewer = {Mark R. Jerrum},
	pages = {346--354},
	publisher = {IEEE},
	title = {The mixing rate of {M}arkov chains, an isoperimetric inequality, and computing the volume},
	url = {https://doi.org/10.1109/FSCS.1990.89553},
	year = {1990}}

@article{Klartag23log,
	author = {Klartag, Bo'az},
	fjournal = {Ars Inveniendi Analytica},
	journal = {Ars Inveniendi Analytica},
	mrclass = {52A40 (58J65)},
	mrnumber = {4603941},
	mrreviewer = {Ge Xiong},
	number = {4},
	pages = {1--17},
	title = {Logarithmic bounds for isoperimetry and slices of convex sets},
	url = {https://mathscinet.ams.org/mathscinet-getitem?mr=4603941},
	year = {2023}}

@article{KLS95isop,
	author = {Kannan, Ravi and Lov\'{a}sz, L\'{a}szl\'{o} and Simonovits, Mikl\'{o}s},
	doi = {10.1007/BF02574061},
	fjournal = {Discrete \& Computational Geometry. An International Journal of Mathematics and Computer Science},
	issn = {0179-5376},
	journal = {Discrete \& Computational Geometry},
	mrclass = {52A40 (52A38 68Q20)},
	mrnumber = {1318794},
	mrreviewer = {Carla Peri},
	number = {3},
	pages = {541--559},
	title = {Isoperimetric problems for convex bodies and a localization lemma},
	url = {https://doi.org/10.1007/BF02574061},
	volume = {13},
	year = {1995}}

@article{Turchin71computation,
	author = {Turchin, Valentin F.},
	journal = {Theory of Probability \& Its Applications},
	number = {4},
	pages = {720-724},
	title = {On the computation of multidimensional integrals by the {M}onte-{C}arlo method},
	volume = {16},
	year = {1971}}

@inproceedings{BG79constraint,
	author = {Boneh, Arnon and Golan, A.},
	booktitle = {Third European congress on operations research (EURO III), Amsterdam},
	title = {Constraints' redundancy and feasible region boundedness by random feasible point generator ({RFPG})},
	year = {1979}}

@incollection{Costabel2017,
	address = {Cham},
	author = {Costabel, Martin},
	booktitle = {Recent Trends in Operator Theory and Partial Differential Equations},
	doi = {10.1007/978-3-319-47079-5_4},
	editor = {Maz'ya, Vladimir and Natroshvili, David and Shargorodsky, Eugene and Wendland, Wolfgang L.},
	pages = {79--88},
	publisher = {Birkh{\"a}user},
	series = {Operator Theory: Advances and Applications},
	title = {Inequalities of {B}abu{\v{s}}ka--{A}ziz and {F}riedrichs--{V}elte for differential forms},
	volume = {258},
	year = {2017}}

@article{HurriSyrjanen1994,
	author = {Hurri-Syrj{\"a}nen, Ritva},
	doi = {10.1090/S0002-9939-1994-1169032-X},
	journal = {Proceedings of the American Mathematical Society},
	number = {1},
	pages = {213--222},
	title = {An improved {P}oincar{\'e} inequality},
	volume = {120},
	year = {1994}}

@article{Zsuppan2020,
	author = {Zsupp{\'a}n, S{\'a}ndor},
	doi = {10.4171/ZAA/1656},
	journal = {Zeitschrift f{\"u}r Analysis und ihre Anwendungen},
	number = {2},
	pages = {171--184},
	title = {Connections between optimal constants in some norm inequalities for differential forms},
	volume = {39},
	year = {2020}}

@article{CW10poincare,
	author = {Seng-Kee Chua and Richard L. Wheeden},
	doi = {10.4310/MRL.2010.v17.n5.a15},
	journal = {Mathematical Research Letters},
	number = {5},
	pages = {993--1011},
	publisher = {International Press of Boston},
	title = {Weighted {P}oincar\'{e} inequalities on convex domains},
	url = {https://doi.org/10.4310/MRL.2010.v17.n5.a15},
	volume = {17},
	year = {2010}}

@article{DD08improved,
	author = {Irene Drelichman and Ricardo G. Dur\'{a}n},
	doi = {https://doi.org/10.1016/j.jmaa.2008.06.005},
	issn = {0022-247X},
	journal = {Journal of Mathematical Analysis and Applications},
	number = {1},
	pages = {286-293},
	title = {Improved {P}oincar\'{e} inequalities with weights},
	url = {https://www.sciencedirect.com/science/article/pii/S0022247X08006276},
	volume = {347},
	year = {2008}}

\appendix

\section{Functional-analytic background\label{sec:Functional-analytic-background}}

This appendix expands the functional-analytic facts used in the proof.
All Sobolev spaces are taken over the open set $\K$. A standard reference
is Evans~\cite[Chapter~5]{evans10partial}.

\paragraph{Hilbert spaces.}

A real \emph{Hilbert space} is a vector space $F$ with an inner product
$\inner{f,g}_{F}$ that is complete in the induced norm $\norm f^{2}_{F}:=\inner{f,f}_{F}$,
where completeness means that every Cauchy sequence converges to an
element of the same space. If $T:F\to G$ is a bounded linear map
between Hilbert spaces, its adjoint is the unique bounded map $T^{*}:G\to F$
satisfying
\[
\inner{Tf,g}_{G}=\inner{f,T^{*}g}_{F}\qquad\text{for every }f\in F,\ g\in G\,.
\]
The Hilbert projection theorem states that every nonempty closed convex
set contains a unique element of minimum norm. 

\paragraph{Weak derivatives.}

For smooth $f$ and $\varphi\in C^{\infty}_{c}(\K)$, integration
by parts gives
\[
\int_{\K}f\,\partial_{j}\varphi\,\D x=-\int_{\K}(\partial_{j}f)\,\varphi\,\D x\,.
\]
There is no boundary term because $\varphi$ is compactly supported
in $\K$. We take this identity as the definition for $f\in L^{2}(\K)$
that need not possess a classical derivative: a function $v_{j}\in L^{2}(\K)$
is the \emph{weak derivative} $\partial_{j}f$ if
\[
\int_{\K}f\,\partial_{j}\varphi\,\D x=-\int_{\K}v_{j}\,\varphi\,\D x\qquad\text{for every }\varphi\in C^{\infty}_{c}(\K)\,.
\]
Thus, the definition preserves precisely the integration-by-parts
rule needed later. The weak derivative is unique (in a.e.\ sense).
It is also stable under $L^{2}$ limits: if $f_{k}\to f$ and $\partial_{j}f_{k}\to v_{j}$
in $L^{2}(\K)$, passing to the limit in the preceding identity shows
that $\partial_{j}f=v_{j}$.

\paragraph{Sobolev spaces.}

The Sobolev space is
\[
H^{1}(\K):=\{f\in L^{2}(\K):\partial_{j}f\in L^{2}(\K)\text{ for every }j\in[n]\}\quad\text{with norm }\norm f^{2}_{H^{1}(\K)}:=\int_{\K}(\abs f^{2}+\abs{\nabla f}^{2})\,\D x\,.
\]
The stability of weak derivatives under limits shows that $H^{1}(\K)$
is complete, hence Hilbert.

Let us consider the $H^{1}$-closure of smooth functions with compact
support in $\K$, $H^{1}_{0}(\K):=\overline{C^{\infty}_{c}(\K)}^{H^{1}(\K)}$.
Namely, for any $f\in H^{1}_{0}(\K)$, there exists a sequence of
compactly-supported smooth functions $\vphi_{k}$ such that $\vphi_{k}\to f$
in the $H^{1}$ norm. This closure encodes a zero boundary condition
without requiring pointwise boundary values.

\paragraph{Actions on functions and measures.}

Let $\mc B(\K)$ be the Borel $\sigma$-algebra of $\K$. A \emph{Markov
kernel} is a map $P:\K\times\mc B(\K)\to[0,1]$ such that $P(x,\cdot)$
is a probability measure for every $x\in\K$ and $x\mapsto P(x,A)$
is measurable for every $A\in\mc B(\K)$. For a bounded measurable
function $f$ and a probability measure $\mu$, define
\[
(Pf)(x):=\int_{\K}f(y)\,P(x,\D y),\qquad(\mu P)(A):=\int_{\K}P(x,A)\,\D\mu(x)\,.
\]
The first action is the associated \emph{Markov operator} on functions,
while the second one is the law after one transition from an initial
law $\mu$.

\paragraph{Stationarity and reversibility.}

A probability measure $\pi$ is \emph{stationary} if $\pi P=\pi$.
In this case $\pi(Pf)=\pi f$, so $P$ preserves $L^{2}_{0}(\pi)$.
Jensen's inequality and stationarity give
\[
\norm{Pf}^{2}_{2}\leq\int_{\K}P(f^{2})\,\D\pi=\int_{\K}f^{2}\,\D\pi=\norm f^{2}_{2}\,,
\]
so $P$ is a bounded linear map. Then, its adjoint $P^{*}$ is characterized
by $\inner{Pf,g}_{\pi}=\inner{f,P^{*}g}_{\pi}$ for every $f,g\in L^{2}(\pi)$.
The operator is \emph{self-adjoint} if $P=P^{*}$. A Markov kernel
$P$ is \emph{reversible} with respect to $\pi$ if $\D\pi(x)\,P(x,\D y)=\D\pi(y)\,P(y,\D x)$.
This is equivalent to self-adjointness of $P$ on $L^{2}(\pi)$.

\section{From Babu\v{s}ka--Aziz constant to improved Poincar\'e constant
\label{sec:BA_to_Poincare}}

Throughout this section, we use $\Omega\subset\Rn$ for the domain,
$Q=L^{2}_{0}(\Omega)$ for the space of square-integrable functions
over $\Omega$ with zero mean and norm $\norm f_{Q}:=\norm f_{2}$,
and $V=H^{1}_{0}(\Omega;\Rn)$ for the space of vector fields on $\Omega$
with square-integrable derivatives that vanish on the boundary and
norm $\norm v_{V}=\norm{\nabla v}_{2}$ (note that this is the $L^{2}$-norm
of the Frobenius norm of $\nabla v\in\Rnn$). As implied in \S\ref{subsec:Digesting-proof-ideas},
we will focus on the divergence operator $\Div:V\to Q$. Since $\int_{\Omega}\Div v=\int_{\partial\Omega}v\cdot\nu=0$
for $v\in V$ (here, $\nu$ is the outward normal vector), this divergence
operator is well-defined.

In this section, we will recap the proof of Theorem~\ref{thm:BA_to_IPI}
for the reader's convenience, focusing on Euclidean domains. Note
that this result follows from combining \cite[Theorem~2.1]{Costabel2017}
and \cite[Lemma~3.4]{Zsuppan2020}. The proof of the theorem will
follow from the following lemmas. The first one is just a dual definition
of $C_{\msf{BA}}$. 
\begin{lem}[Duality]
\label{lem:dual_BA} The following are equivalent for any $C>0$: 
\begin{enumerate}
\item For any $f\in Q$, there exists $v\in V$ such that $\Div v=f$ and
$\norm{\nabla v}^{2}_{2}\le C\,\norm f^{2}_{2}$.
\item For any $f\in Q$, it holds that $\norm f^{2}_{2}\le C\,(\sup_{v\in V,v\neq0}\frac{\inner{f,\Div v}_{Q}}{\norm{\nabla v}_{2}})^{2}$.
\end{enumerate}
\end{lem}

When solving $\Div v=f$, the Babu\v{s}ka--Aziz constant essentially
asks how large $\norm v_{V}$ is compared with $\norm f_{Q}$, and
this is the first statement (i.e., $\norm v_{V}\leq C^{1/2}\,\norm{\Div v}_{Q}$
for some $v\in V$). Consider the adjoint $\Div^{*}:Q\to V$ such
that $\inner{f,\Div v}_{Q}=\inner{\Div^{*}f,v}_{V}$. Then, the RHS
in the second statement is equivalent to $\norm f_{Q}\leq C^{1/2}\,\norm{\Div^{*}f}_{V}$.

The next lemma derives an explicit solution of the supremum on the
RHS in the previous lemma using the Riesz representation theorem.
\begin{lem}[Map to vector field]
\label{lem:Riesz} For any $f\in Q$, there is a unique $u\in V$
such that $\inner{f,\Div v}_{Q}=\inner{u,v}_{V}$ for any $v\in V$,
and 
\[
\norm{\nabla u}_{2}=\sup_{v\in V,v\neq0}\frac{\left\langle f,\Div v\right\rangle _{Q}}{\norm{\nabla v}_{2}}\,.
\]
\end{lem}

\begin{proof}
The map $v\mapsto\left\langle f,\Div v\right\rangle $ is a bounded
linear map, since $\abs{\inner{f,\Div v}}\le\norm f_{2}\,\norm{\Div v}_{2}\le\sqrt{n}\,\norm f_{2}\norm{\nabla v}_{2}$.
By the Riesz representation theorem, there is a unique $u\in V$ such
that 
\[
\inner{f,\Div v}_{Q}=\inner{u,v}_{V}=\inner{\nabla u,\nabla v}\le\norm{\nabla u}_{2}\,\norm{\nabla v}_{2}\,.
\]
Hence, the supremum on the RHS in the lemma statement is at most $\norm{\nabla u}_{2}$.
Setting $v=u$ shows equality. 
\end{proof}

Due to this lemma, it now suffices to show $\norm f_{Q}\le C^{1/2}\,\norm u_{V}$
for some $C$. The next lemma is a calculus identity about divergence
of vector fields.
\begin{lem}
\label{lem:dual_matrix} For $u\in V$, define the $n\times n$ matrix
function $G$ as $G_{ij}:=\partial_{j}u_{i}-\partial_{i}u_{j}$ with
$\norm G^{2}:=\int_{\Omega}\sum_{i<j}G^{2}_{ij}\,\D x$. Then, $\norm{\nabla u}^{2}_{2}=\norm{\Div u}^{2}_{2}+\norm G^{2}$.
\end{lem}

\begin{proof}
Using integration by parts twice and the fact that partial derivatives
commute, 
\[
\norm{\Div u}^{2}_{2}=\sum_{ij}\int_{\Omega}\de_{i}u_{i}\de_{j}u_{j}=-\sum_{ij}\int_{\Omega}u_{i}\partial_{i}\partial_{j}u_{j}=\sum_{ij}-\int_{\Omega}u_{i}\partial_{j}\partial_{i}u_{j}=\sum_{ij}\int_{\Omega}\partial_{j}u_{i}\partial_{i}u_{j}\,.
\]
Let us now compute $\norm G^{2}$:
\begin{align*}
\norm G^{2} & =\sum_{i<j}\int_{\Omega}(\partial_{j}u_{i}-\partial_{i}u_{j})^{2}=\frac{1}{2}\sum_{ij}\int_{\Omega}(\partial_{j}u_{i}-\partial_{i}u_{j})^{2}=\sum_{ij}\int_{\Omega}(\partial_{i}u_{j})^{2}-\sum_{ij}\int_{\Omega}\partial_{j}u_{i}\partial_{i}u_{j}\\
 & =\norm{\nabla u}^{2}_{2}-\norm{\Div u}^{2}_{2}\,,
\end{align*}
which completes the proof.
\end{proof}

Next, using this definition of $G$, we get useful identities.
\begin{lem}
\label{lem:identities} For $f\in Q$, let $u\in V$ be the vector
field given by Lemma~\ref{lem:Riesz} and $G$ be the matrix function
for $u$ defined in Lemma~\ref{lem:dual_matrix}. Then for $h:=f-\Div u$,
we have
\begin{enumerate}
\item $\int_{\Omega}h=0$;
\item $\nabla h=\Div G$ (i.e., $\partial_{i}h=\sum^{n}_{j=1}\partial_{j}G_{ij}$
for $i\in[n]$);
\item $\inner{h,\Div u}_{Q}=\norm G^{2}$.
\end{enumerate}
\end{lem}

\begin{proof}
Since $u=0$ on $\partial\Omega$, we have $\int_{\Omega}\Div u=\int_{\partial\Omega}\inner{u,\nu}=0$.
Hence, $\int_{\Omega}h=\int_{\Omega}(f-\Div u)=0$. 

Next, using weak derivatives and integration by parts, we can deduce
$\Delta u=\nabla f$ from $\inner{f,\Div v}_{Q}=\inner{u,v}_{V}$.
Using this, for each $i$,
\[
\sum_{j}\partial_{j}G_{ij}=\sum_{j}\partial_{j}(\partial_{j}u_{i}-\partial_{i}u_{j})=\Delta u_{i}-\partial_{i}\Div u=\partial_{i}f-\partial_{i}\Div u=\partial_{i}h\,.
\]
Taking $v=u$ in the definition of $u$, we obtain $\inner{f,\Div u}_{Q}=\norm{\nabla u}^{2}_{2}=\norm{\Div u}^{2}_{2}+\norm G^{2}$,
and this implies that $\inner{h,\Div u}_{Q}=\inner{f-\Div u,\Div u}_{Q}=\norm G^{2}$.
\end{proof}

The following lemma captures the heart of the argument. Recall that
$\delta(x)$ is the distance of $x$ to the boundary of $\Omega$. 
\begin{lem}
\label{lem:poincare_bound} For any $h\in Q$ and $G$ an anti-symmetric
matrix function satisfying $\nabla h=\Div G$, we have $\norm{\delta\nabla h}_{2}\le2\norm G.$
Hence, if $\Omega$ has finite improved Poincar\'e constant $\cipi$,
then $\norm h^{2}_{2}\le4\cipi\,\norm G^{2}$. 
\end{lem}

\begin{proof}
By a standard regularization argument, it suffices to prove the claim
for smooth $h$ and $G$; the general case follows by approximation.
We start with the relation between $h$ and $G$. Using $\nabla h=\Div G$
in $(i)$ below,
\begin{align*}
\norm{\delta\nabla h}^{2}_{2} & =\int_{\Omega}\delta^{2}\,\abs{\nabla h}^{2}=\sum_{i}\int_{\Omega}\delta^{2}\,\partial_{i}h\partial_{i}h\underset{(i)}{=}\sum_{ij}\int_{\Omega}\delta^{2}\,\partial_{i}h\partial_{j}G_{ij}\underset{(ii)}{=}-\sum_{ij}\int_{\Omega}\partial_{j}(\delta^{2}\partial_{i}h)\,G_{ij}\\
 & =-2\sum_{ij}\int_{\Omega}\delta\,(\partial_{j}\delta)(\partial_{i}h)\,G_{ij}-\sum_{ij}\int_{\Omega}\delta^{2}\,(\partial_{ij}h)\,G_{ij}\underset{(iii)}{=}-2\sum_{ij}\int_{\Omega}\delta\,(\partial_{j}\delta)(\partial_{i}h)\,G_{ij}\,,
\end{align*}
where in $(ii)$, we used integration by parts and the fact that $\delta=0$
on $\partial\Omega$, and in $(iii)$, we used $\inner{\hess h,G}=0$
since the Hessian of $h$ is symmetric, and $G$ is anti-symmetric.
Then,
\begin{align*}
\norm{\delta\nabla h}^{2}_{2} & =-2\int_{\Omega}\delta\sum_{i<j}\{(\partial_{j}\delta)(\partial_{i}h)-(\partial_{i}\delta)(\partial_{j}h)\}\,G_{ij}\le2\int_{\Omega}\delta\,\Babs{\sum_{i<j}\{(\partial_{j}\delta)(\partial_{i}h)-(\partial_{i}\delta)(\partial_{j}h)\}\,G_{ij}}\\
 & \le2\int_{\Omega}\delta\,\Bpar{\sum_{i<j}\{(\partial_{j}\delta)(\partial_{i}h)-(\partial_{i}\delta)(\partial_{j}h)\}^{2}}^{1/2}\,\Bpar{\sum_{i<j}G^{2}_{ij}}^{1/2}\,.
\end{align*}
Next, the summation in the first term is equal to $\abs{\nabla\delta}^{2}\abs{\nabla h}^{2}-\inner{\nabla\delta,\nabla h}^{2}$
and thus at most $\abs{\nabla\delta}^{2}\abs{\nabla h}^{2}\le\abs{\nabla h}^{2}$,
since $\delta$ is $1$-Lipschitz. Therefore, 
\[
\norm{\delta\nabla h}^{2}_{2}\le2\int_{\Omega}\delta\,\abs{\nabla h}\,\Bpar{\sum_{i<j}G^{2}_{ij}}^{1/2}\le2\,\Bpar{\int_{\Omega}\delta^{2}\,\abs{\nabla h}^{2}}^{1/2}\Bpar{\int_{\Omega}\sum_{i<j}G^{2}_{ij}}^{1/2}\,,
\]
which completes the proof.
\end{proof}

With these lemmas in hand, we can prove the main theorem of this section.
\begin{proof}[Proof of Theorem \ref{thm:BA_to_IPI}.]
 Let $f\in Q$. Pick $u\in V$ from Lemma~\ref{lem:Riesz} and $G$
from Lemma~\ref{lem:dual_matrix} for this $u$. Then for $h=f-\Div u$,
we have $\norm h^{2}_{2}\le4\cipi\,\norm G^{2}$ by Lemma~\ref{lem:poincare_bound}.
Moreover, using Lemma~\ref{lem:identities} and Cauchy--Schwarz,
\[
\norm G^{2}=\inner{h,\Div u}_{Q}\le\norm h_{2}\,\norm{\Div u}_{2}\le2\sqrt{\cipi}\,\norm G\,\norm{\Div u}_{2}\,,
\]
which implies $\norm G^{2}\le4\cipi\,\norm{\Div u}^{2}_{2}$. Finally,
\begin{align*}
\norm f^{2}_{2} & =\norm h^{2}_{2}+\norm{\Div u}^{2}_{2}+2\inner{h,\Div u}_{Q}\le4\cipi\,\norm G^{2}+\norm{\Div u}^{2}_{2}+2\norm G^{2}\\
 & \le4\cipi\,\norm G^{2}+\norm{\Div u}^{2}_{2}+\norm G^{2}+4\cipi\,\norm{\Div u}^{2}_{2}=(1+4\cipi)\,(\norm{\Div u}^{2}_{2}+\norm G^{2})\\
 & =(1+4\cipi)\norm{\nabla u}^{2}_{2}=(1+4\cipi)\,\bpar{\sup_{v\in V,v\neq0}\frac{\inner{f,\Div v}_{Q}}{\norm{\nabla v}_{2}}}^{2}
\end{align*}
which proves the theorem. 
\end{proof}

\end{document}